\documentclass[10pt,journal]{IEEEtran}

\usepackage{cite}
\usepackage{amsmath,amssymb,amsfonts,amsthm}
\usepackage{algorithmic}
\usepackage[algo2e]{algorithm2e}
\usepackage{graphicx}
\usepackage{textcomp}
\usepackage{xcolor}
\usepackage{booktabs}
\usepackage[hidelinks]{hyperref}
\usepackage{algorithm}
\usepackage{subcaption}
\allowdisplaybreaks
\DeclareMathOperator*{\tsum}{\textstyle\sum}
\newtheorem{proposition}{Proposition}

\newtheorem{lemma}{Lemma}
\newtheorem{theorem}{Theorem}

\newtheorem{assumption}{Assumption}

\renewcommand{\v}[1]{\ensuremath{\boldsymbol{#1}}}

\let\sss= \scriptscriptstyle

\begin{document}

\title{
{Adaptive Differentially Quantized Subspace Perturbation (ADQSP): A Unified Framework for Privacy-Preserving Distributed Average Consensus}
\author{Qiongxiu Li, Jaron Skovsted Gundersen, Milan Lopuha\"a-Zwakenberg and Richard Heusdens

 \thanks{Q. Li is with Tsinghua University, China (email: qiongxiu\allowbreak li@mail.tsinghua.edu.cn). }
 \thanks{J. S. Gundersen is with the Department of Electronic Systems, Aalborg University, Denmark, (email:jaron@es.aau.dk).}
 \thanks{M. Lopuha\"a-Zwakenberg is with University of Twente, the Netherlands, (email:m.a.lopuhaa@utwente.nl).}
 \thanks{R. Heusdens is with the Netherlands Defence Academy (NLDA), the Netherlands, and with the Faculty of Electrical Engineering, Mathematics and Computer Science, Delft University of Technology, Delft, the Netherlands (email: r.heusdens@\{mindef.nl,tudelft.nl\}).}
}

}

\maketitle
\raggedbottom

\addtolength{\abovedisplayskip}{-1.0mm}
\addtolength{\belowdisplayskip}{-1.0mm}
\begin{abstract}
Privacy-preserving distributed average consensus has received significant attention recently due to its wide applicability.  Based on the achieved performances, existing approaches can be broadly classified into perfect accuracy-prioritized approaches such as secure multiparty computation (SMPC), and worst-case privacy-prioritized approaches such as differential privacy (DP). Methods of the first class achieve perfect output accuracy but reveal some private information, while methods from the second class provide privacy against the strongest adversary at the cost of a loss of accuracy. 
In this paper, we propose a general approach named adaptive differentially quantized subspace perturbation (ADQSP) which combines quantization schemes with so-called subspace perturbation. Although not relying on cryptographic primitives, the proposed approach enjoys the benefits of both accuracy-prioritized and privacy-prioritized methods and is able to unify them. More specifically, we show that by varying a single quantization parameter the proposed method can vary between SMPC-type performances and DP-type performances. Our results show the potential of exploiting traditional distributed signal processing tools for providing cryptographic guarantees. In addition to a comprehensive theoretical analysis, numerical validations are conducted to substantiate our results. 
\end{abstract}


\section{Introduction}

As the world is becoming increasingly interconnected and digitized, data are often collected and stored in personal devices such as tablets and phones \cite{poushter2016smartphone}. To process such massive amounts of data over different devices poses many challenges including: (1) the requirement for distributed processing tools that are able to process data in a network of devices without any centralized coordination; (2) the need for lightweight solutions as these devices are often limited in computational resources, 
and (3) the demand for privacy-preserving algorithms, 
as these devices often contain sensitive personal data captured by sensors such as images and GPS data \cite{poushter2016smartphone}. Combined, these challenges call for interdisciplinary research across different fields such as cryptography, information theory and distributed signal processing for developing efficient and privacy-preserving distributed computation methods. 

The average consensus problem has been intensively investigated and widely applied in many applications, e.g., optimization \cite{olfati2007consensus}, group coordination \cite{tsitsiklis1984problems}, and federated learning \cite{li2020federated}. Recently, privacy-preserving distributed average consensus, which aims to guarantee all participants in a network to output an accurate average consensus without violating privacy concerns, has received a lot of attention \cite{hendriks2013privacy,ruan2017secure,braca2016learning,hale2015differentially,hale2018cloud, nozari2017differentially,kefayati2007secure,huang2012differentially,guo2018practical, gupta2017privacy,gupta2019statistical,li2019privacyA,li2019privacyS,Jane2020ICASSP,Jane2020TSP,manitara2013privacy,mo2017privacy,he2019privacy,xiong2022privacy,zhou2022private}. 
There are at least five aspects to be considered when designing privacy-preserving distributed average consensus algorithms: centralized vs.\ decentralized coordination, computational complexity, communication costs, and most importantly, the achieved privacy level and output accuracy. 
With respect to the first concern, in this paper we only consider fully distributed/decentralized solutions which do not require any centralized coordination such as a trusted third party. 
Secondly, the computational complexity depends heavily on the so-called security model, i.e., whether we demand computational or information-theoretical security. 
In this work, we will restrict ourselves to information-theoretical security based algorithms since they are generally computationally less complex compared to computational security based approaches as the former does not involve complex encryption operations. 

With respect to communication cost, privacy, and accuracy, most existing work focuses on privacy and accuracy, while the aspect of communication cost is rarely addressed. Ideally, one would like perfect accuracy while guaranteeing privacy even against the strongest adversary, which is all participants except one collude. This, however, has been shown to be impossible \cite{Jane2020TIFS}. The main reason is that having knowledge of the average result will reveal the private data held by the only non-colluding participant.
As a result, existing work falls into two categories. The first category contains algorithms that guarantee privacy against the worst-case adversary using differential privacy (DP) techniques \cite{dwork2006,dwork2006calibrating} by inserting noise into the calculations \cite{nozari2017differentially,kefayati2007secure,huang2012differentially,guo2018practical}. This guarantees one individual's privacy even when all other participants are compromised, at the cost of output accuracy. The second category contains secure multiparty computation (SMPC) based approaches like secret sharing \cite{gupta2017privacy,gupta2019statistical,li2019privacyA,li2019privacyS} and correlated noise insertion \cite{manitara2013privacy,mo2017privacy,he2019privacy}, which offer perfect accuracy but only guarantee privacy under additional assumptions on the set of compromised participants.


\subsection{Paper contribution}

In this paper, we present a communication efficient algorithm that links the performances of SMPC and DP approaches in a unified framework. Our key results are summarized as follows:
\begin{enumerate}
  \item We analyze the problem of privacy-preserving distributed average consensus over networks and prove two impossibility results about the performances of privacy and accuracy given the network topology. In particular, we prove that the so-called ideal world defined in SMPC is impossible to achieve if the adversary disconnects the network. 
  \item We propose a novel approach, referred to as \emph{adaptive differentially quantized subspace perturbation (ADQSP)}, by combining the potential of quantization schemes and the subspace perturbation technique. It enjoys superior performance compared to existing approaches in terms of accuracy, privacy and communication efficiency. 
  \item We show that our proposed approach can achieve SMPC and (relaxed) DP performances by appropriate parameter settings. To the best of our knowledge, the proposed approach is the first information-theoretical approach that unifies both SMPC and DP in this context.
\end{enumerate}
Both theoretical investigations and numerical validations are presented to consolidate our claims. We published preliminary results in \cite{Jane2022twoforone} where only the connection of quantized subspace perturbation and DP is shown. In this paper we show that the proposed ADQSP approach can achieve both SMPC and DP related performances. In addition, we give a complete information-theoretical analysis of privacy via  mutual information. 

\subsection{Outline and Notation}
The paper is organized as follows. In Section \ref{sec.pre}, we introduce the problem setup and define the adversary model and privacy metrics. In Section \ref{sec.imRes}, we analyze the problem at hand and prove two impossibility results regarding the performances. In Section \ref{sec.prop}, we introduce the proposed approach. In Section \ref{sec.smpc_dp}, we analyze the performance of existing SMPC and DP approaches, while in Section \ref{sec.prop2exist}, we analyze the performance of the proposed approach and show that by tuning a single parameter the performance can vary between that of SMPC and (relaxed) DP.
Numerical validations are demonstrated in Section \ref{sec.numRes} and conclusions are drawn in Section \ref{sec.conclu}.

We use bold letters for vectors ($\v x$) and matrices ($\v X$), where capital letters are used for matrices. Sets are denoted by calligraphic letters ($\mathcal{X}$). $\v I$ denotes the identity matrix. We refer to the $i$-th entry of a vector $\v x$ as $x_i$.\footnote{For simplicity we assume that the entries $x_i$ are scalar variables but the results can easily be generalized to arbitrary dimensions.}  With a slight abuse of notation, we denote a random variable by a capital letter $X$ no matter if the outcome is a scalar $x$, a vector $\v x$, or a matrix $\v X$. 
%

\section{Preliminaries} \label{sec.pre}
In this section we review some necessary fundamentals for the remainder of the paper.
\subsection{Privacy-preserving distributed average consensus}\label{subsec.network} 
A network is modelled as a graph $\mathcal{G} = ({\cal V},\mathcal{E})$, where ${\cal V} = \{1,\ldots,n\}$ is the set of vertices representing the nodes/agents/participants in the network, and $\mathcal{E} \subseteq {\cal V}\times {\cal V}$ is the set of $m := |\mathcal{E}|$ (undirected) edges, representing the communication links in the network. The set of neighbors of node $i$ is denoted as ${\cal N}_i=\{j\,| \,(i,j)\in \mathcal{E}\}$ and its degree is $d_i=|{\cal N}_i|$. A graph is connected if for every pair of nodes there is a path between them.
In this paper, for the feasibility of the average output, we assume the graph $\mathcal{G}$ is connected.

Assume each node $i\in {\cal V}$ in the network holds private data $s_i$. The goal of privacy-preserving distributed average consensus is to allow each node to obtain the average of all private data over the network, i.e., 
\begin{align}\label{eq:setup}
  s_{\sss \mathrm{ave}}=\frac{1}{n}\sum_{\sss i\in {\cal V}}s_i,
\end{align}
without revealing each node's private data and without any centralized coordination.

\subsection{Adversary model}\label{subsec.adv}
In this paper we consider both the well-known eavesdropping adversary and the so-called passive (i.e., honest-but-curious) adversary. Throughout the paper, we assume these two adversaries can cooperate by sharing information to increase the chance of inferring private data. The eavesdropping adversary works by eavesdropping all communication channels (edges) between nodes in the network.
The passive adversary is assumed to be able to corrupt a subset of the nodes, referred to as {\em corrupt nodes}. The nodes that are not corrupt, on the other hand, are referred to as {\em honest nodes}. The corrupt nodes are assumed to not deviate from the protocol, but collect all information they see throughout the protocol. The collected information includes for example the inputs of the corrupt nodes and the messages they receive during the computation from neighboring nodes. Denote the set of honest nodes as ${\cal V}_h$ and the corrupt nodes as ${\cal V}_c$. We let $\mathcal{E}_h = \{(i,j)\in \mathcal{E}\mid i,j\in {\cal V}_h\}$ be the set of edges between two honest nodes and $\mathcal{E}_c=\mathcal{E}\setminus \mathcal{E}_h$ be the set of edges containing at least one corrupt node. Furthermore, the set of honest neighbors of node $i$ is denoted by ${\cal N}_{\sss i,h}=\{j\, |\, (i,j)\in \mathcal{E}_h\}$.

\subsection{Performance evaluation and corresponding metrics}\label{subsec.metric}
To evaluate the performance of a privacy-preserving algorithm, there are at least two key requirements.
\subsubsection{Output accuracy}
The output accuracy measures how good the functionality of the algorithm is. A typical way to quantify the output accuracy is to adopt some distance measures to evaluate ``how far'' the output of the privacy-preserving algorithm is, denoted as $\tilde{\v y}\in \mathbb{R}^{n}$, to the desired output $\v y\in \mathbb{R}^{n}$. 
In this paper we use the mean squared error (MSE) to quantify the output accuracy as it is widely used for iterative processes. It is defined as:
\begin{align}
  e=\frac{1}{n} \left|\left|\tilde{\v y}-\v y \right|\right|^{2}. \label{eq:MSE}
\end{align}

\subsubsection{Individual privacy}\label{sec:MI}
Individual privacy measures how well the private data of an honest node is being protected against the considered adversaries. Since the approaches considered here are information-theoretical ones, privacy should also be quantified using an information-theoretical measure. Widely used information-theoretical metrics for assessing privacy loss include for example  
$\epsilon$-differential privacy and mutual information \cite{cover2012elements}.  In this paper we adopt mutual information as our primary privacy metric for two main reasons. Firstly,  its efficacy has been demonstrated in the realm of privacy-preserving distributed processing and in various applications \cite{Jane2020TIFS,bu2020tightening}. Secondly, it has been shown in \cite{cuff2016differential,barthe2011information}  that mutual information is a relaxed version of $\epsilon$-differential privacy (thus in Section \ref{sec.smpc_dp} we refer that our proposed approach can achieve performances of relaxed DP approaches). And it is easier to implement mutual information in practice \cite{gotz2011publishing}. 
 Considering 
two (continuous) random variables $X$ and $Y$, the mutual information is defined by
\begin{align}\label{eq:MI_def}
  {\rm I}(X;Y)=h(X)-h(X|Y),
\end{align}
where $h(X)$ denotes the differential entropy  of $X$ and $h(X|Y)$ is the conditional differential entropy, assuming they exist.\footnote{In the case of discrete random variables we replace both differential entropies by the Shannon entropy so that $I(X;Y) = H(X) - H(Y|X)$.} We have that ${\rm I}(X;Y)=0$ when $X$ and $Y$ are independent, meaning that $Y$ does not carry any information about $X$ and ${\rm I}(X;Y)$ is maximal when there is a one-to-one relation between $Y$ and $X$.

\section{Problem analysis and impossibility results} \label{sec.imRes}
In this section we prove two impossibility results, showing that in distributed average consensus 1) the twin goals of perfect accuracy and worse-case privacy guarantee cannot be achieved simultaneously; 2) if the subgraph of honest nodes is disconnected, then no protocol can achieve the performances of the case where a trusted third party (TTP) is assumed to be available. These two results are related to the performances of two well-established privacy-preserving techniques: DP and SMPC, respectively. Without loss of generality, throughout the paper, we assume the following (recall the private data $s_i$ is a realization of random variable $S_i$).
\begin{assumption}\label{asu.1}
All private data are statistically independent, i.e., $\forall i,j \in {\cal V}, i\neq j: {\rm I}(S_i;S_j)=0$. 
\end{assumption}
 
\begin{assumption}\label{asu.2}
The size of the network $n=|{\cal V}|$ is known to all nodes. Hence, knowing the average of the data implies knowing the total sum and vice versa. 
\end{assumption}

\noindent
We have the following results. 
\subsection{Impossibility result I: Perfect accuracy and worst-case privacy cannot be achieved simultaneously} 
\label{subsec.imDP}
\begin{lemma}
By assuming perfect accuracy and worst-case privacy, i.e., ${\cal V}_c={\cal V}\setminus\{i\}$, the individual privacy leakage becomes 
\begin{align}\label{eq:dptrade}
{\rm I}(S_i;\tsum_{\sss j\in {\cal V}} S_j,\{S_j\}_{\sss j\in {\cal V}_c})={\rm I}(S_i;S_i),
\end{align} 
which is maximal.
\end{lemma}
The proof follows trivially from the fact that $\tsum_{\sss j\in {\cal V}} S_j - \tsum_{\sss j\in {\cal V}_c} S_j = S_i$ \cite{Jane2020TIFS}. Similar results have also been proved via different metrics such as differential privacy \cite{nozari2017differentially}.
Thus, if there are $n-1$ corrupt nodes, it is impossible to achieve any privacy for the only remaining honest node without compromising accuracy. Hence, there is an inherent trade-off between privacy and accuracy for differential privacy related approaches. 


\subsection{Impossibility result II: An SMPC-like ideal world cannot be achieved unless all honest nodes are connected} \label{subsec.smpcImp}

In SMPC, the \emph{ideal world} describes the setting in which a TTP is assumed to be available. The TTP will (securely) collect all private data from each node and then compute the output of the function and send back the output to the nodes. However, in practice a TTP might not exist. The goal of SMPC is to design a protocol to replace a TTP. In the context of average consensus, as the corrupt nodes have available both their own private data and the average result, the information observed by the adversary in an ideal world setting is given by:
\begin{align}\label{eq:videal}
  {\cal O}_{\sss \mathrm{ideal},{\cal V}_c}&=\{S_i\}_{\sss i\in{\cal V}_c}\cup \{\tsum_{\sss i\in {\cal V}}S_i\}.
\end{align}
Hence, for an arbitrary honest node $i\in{\cal V}_h$ we have 
\begin{align}\label{eq:ideal_inf}
{\rm I}(S_i;{\cal O}_{\sss \mathrm{ideal},{\cal V}_c})={\rm I}(S_i;\tsum_{\sss i\in {\cal V}_h}S_i).
\end{align}
Hence, a perfect SMPC protocol for privacy-preserving distributed average consensus outputs the correct average value and guarantees that the adversary obtains no more information about the honest nodes' private data than \eqref{eq:ideal_inf}. The following result shows that such a protocol cannot exist when the honest nodes do not form a connected subgraph.

\begin{proposition}\label{prop:smpc}
Let $\mathcal{G}_h=({\cal V}_h,\mathcal{E}_h)$ be the subgraph of $\mathcal{G}$ obtained by removing all corrupt nodes. 
Let $\mathcal{G}_{\sss h,1},\ldots,\mathcal{G}_{\sss h,k_h}$ be the components (connected subgraphs that are not part of any larger connected subgraph) of $\mathcal{G}_h$ and let ${\cal V}_{\sss h,k}$ denote the vertex set of $\mathcal{G}_{\sss h,k}$.
Then for any protocol outputting $f(s_1,s_2,\ldots,s_n)=\tsum_{\sss i=1}^n s_i$ to all nodes, after the protocol the adversary will always have learned $\{\tsum_{\sss i\in {\cal V}_{\sss h,k}}S_i\}_{\sss k=1,2,\ldots,k_h}$.
\end{proposition}
\begin{proof}
See Appendix~\ref{pf.smpc}.
\end{proof}

In other words, any algorithm with perfect accuracy will always leak the partial sums of the components of $\mathcal{G}_h$ to the adversary. Note that $\mathcal{G}_h$ is very likely to be disconnected in incomplete networks, especially in the presence of a lot of corrupt nodes. Hence, a perfect SMPC protocol for privacy-preserving distributed average consensus does not exist if the adversary disconnects the honest nodes \cite{Beimel2005}.

\subsection{Redefined ideal world for incomplete networks} \label{subsec.reideal}
Using Proposition \ref{prop:smpc}, we can derive a tighter bound on the information view of the adversary by taking the network topology into account, which we refer to as the redefined ideal world. The information observed by the adversary in this redefined ideal world setting, referred to as ${\cal O}_{\sss \mathrm{reideal}}$, is 
\begin{align}\label{eq:videal_new}
   {\cal O}_{\sss \mathrm{reideal},{\cal V}_c} =\{S_j\}_{\sss j\in{\cal V}_c}\cup\{\tsum_{\sss i\in {\cal V}_{\sss h,k}}S_i\}_{\sss k=1,2,\ldots,k_h}.
\end{align}
Hence, since ${\cal V}_{\sss h,k}\subseteq {\cal V}_h$ we have ${\rm I}(S_i; {\cal O}_{\sss \mathrm{reideal},{\cal V}_c}) \geq {\rm I}(S_i; {\cal O}_{\sss \mathrm{ideal},{\cal V}_c})$ and equality holds if and only if the subgraph $\mathcal{G}_h$ is connected, i.e., if $k_h=1$. 

For the simplicity of notation, assume that the honest node $i$ belongs to the first honest component, i.e., $i\in {\cal V}_{\sss h,1}$. We conclude that a perfect SMPC protocol should reveal no more information about $s_i$ than
\begin{align} \label{eq:reidealSMPC}
  &{\rm I}(S_i; {\cal O}_{\sss \mathrm{reideal},{\cal V}_c})={\rm I}(S_i;\tsum_{\sss j \in {\cal V}_{\sss h,1}} S_j).
\end{align}


\section{Proposed ADQSP protocol} \label{sec.prop}
The proposed ADQSP is constructed from two building blocks in distributed optimization, namely \emph{adaptive differential quantization} \cite{jane2022elsevier,Jane2022twoforone} and \emph{subspace perturbation} \cite{Jane2020ICASSP,Jane2020TSP}. 
In what follows, we first briefly introduce the fundamentals of distributed optimization and these two building blocks, and afterwards we explain the details of the proposed approach.

\subsection{Distributed optimization}\label{subsec.disOptim}
To solve \eqref{eq:setup} in a distributed manner, we first formulate it as a constrained optimization problem:
\begin{equation} \label{eq:setupAve}
\begin{array}{ll}{\displaystyle \min_{\{\sss x_1,\ldots,x_n\}}} & {{\displaystyle\tsum_{\sss i \in {\cal V}} \frac{1}{2}\| x_i - s_i\|^2}} \\ {\text { s.t. }} & {\forall} (i,j)\in \mathcal{E}:~ { x_i= x_j}, \end{array}
\end{equation}
where $x_i$ is a local optimization variable at node $i$ with
optimal value $x_i^*= s_{\sss \rm ave}$.
A typical way to solve the above problem is to apply a solver like ADMM \cite{boyd2011distributed} or PDMM \cite{zhang2018distributed, sherson2018derivation}. As shown in \cite{sherson2018derivation} these distributed optimization algorithms can be described in a general framework using monotone operator theory \cite{ryu2016primer} and operator splitting techniques. 
For average consensus, the local updating functions are given by
\begin{align}\label{eq:xup}
  & x_{\sss i}^{(t+1)} =\frac{ s_i- \tsum_{\sss j \in {\cal N}_i} B_{\sss i|j} z_{\sss i|j}^{(t)}}{1+cd_i} \\
  & z_{\sss j|i}^{(t+1)}=\theta z_{\sss j|i}^{(t)} +(1-\theta)\left( z_{\sss i|j}^{(t)}+2c B_{\sss i|j} x_i^{(t+1)}\right).
  \label{eq:zup}
\end{align}
Here $\v z \in \mathbb{R}^{2m}$ is an auxiliary variable having entries indicated by $z_{\sss i|j}$ and $z_{\sss j|i}$, held by node $i$ and $j$, respectively, related to edge $(i,j) \in \mathcal{E}$. 
The matrix $\v B \in \mathbb{R}^{m \times n}$ is the graph incidence matrix and $ B_{\sss i|j}$ and $ B_{\sss j|i}$ are edge-related (scalar) weights, which are related to entries of $\v B$ as $ B_{\sss k,i}= B_{\sss i|j}$ and $ B_{\sss k,j}= B_{\sss j|i}$ for $e_k = (i,j)$, i.e., the $k$-th edge in the graph.  We will use the convention that 
$ B_{\sss i|j}=1$ and $ B_{\sss j|i}=-1$ if $i<j$.
In addition, $c>0$ and $\theta \in [0,1)$ are constants, with $c$ controlling the convergence rate and $\theta$ controlling the averaging of the (non-expansive) operators (recall that $d_i=|\mathcal{N}_i|$ is the degree of node $i$). The case $\theta = 0$ corresponds to PDMM and $\theta = 0.5$ corresponds to ADMM.

The protocol is summarized in Algorithm \ref{alg:ave} where $t_{\max}$ denotes the maximum number of iterations. Note that there is also an alternative broadcast algorithm that requires only broadcasting of $x_i^{(t+1)}$'s to all neighbors, instead of exchanging the $z_{j|i}^{(t+1)}$'s one by one; the reception of $x_i^{(t+1)}$ is all that is needed for node $j\in{\cal N}_i$ to compute $z_{j|i}^{(t+1)}$ (see Appendix \ref{appen.broad} for details). 

\begin{algorithm}[t]
\caption{Distributed optimization for average consensus}
\label{alg:ave}
  At each $i \in {\cal V}$:
  \begin{enumerate}
 \item  Initialize $ z_{  i|j}^{(0)}=0$;
		\item For $t=0,1,\ldots,t_{\max}-1$ do
		\begin{enumerate}
		  \item Compute $ x_{  i}^{(t+1)}$ and $\{z_{  j|i}^{(t+1)}\}_{j \in {\cal N}_{  i}}$ using \eqref{eq:xup} and \eqref{eq:zup}, respectively;
		  \item Send $z_{  j|i}^{(t+1)}$ to neighbor $j\in {\cal N}_i$;
		\end{enumerate}
		\item \label{step.out} Output $ x_i^{(t_{\max})}$
  \end{enumerate}
\end{algorithm}

\subsection{Subspace perturbation}\label{subsec:Fund_Subspace}
By inspection of \eqref{eq:xup},  we can see that there can be two noise-insertion options to protect the private data $s_i$: adding noise to the optimization variable $\v x\in \mathbb{R}^{n}$ directly or adding noise to the auxiliary variable $\v z\in \mathbb{R}^{2m}$. The subspace perturbation approach adds noise to the auxiliary variable $\v z$, as it often has more degrees of freedom \cite{Jane2020ICASSP,Jane2020TSP}. We have shown in previous work \cite{Jane2020TSP,jane2022elsevier} that only a part of $\v z$ is predictable (like $\v x$ converges to a value known by the adversary), while the other part is not. To explain this in more detail, 
we first write the local functions \eqref{eq:xup} and \eqref{eq:zup} in a compact form as follows: 
\begin{align}
  \v x^{(t+1)}&=(\v I+c\v C^{\top}\v C)^{-1}(\v s-\v C^{\top}\v z^{(t)}) \label{eq:xupc}\\
  \v z^{(t+1)}&=\theta \v z^{(t)} +(1-\theta)\left(\v P\v z^{(t)}+2c\v P\v C \v x^{(t+1)}\right), \label{eq:zupc}
\end{align}
where $\v C=[\v B_{\sss +}^{\top},\v B_{\sss -}^{\top}]^{\top}\in \mathbb{R}^{2m\times n}$ and $\v B_{\sss +}$ and $\v B_{\sss -}$ contains the positive and negative entries of $\v B$, respectively. Furthermore, $\v P\in \mathbb{R}^{2m\times 2m}$ is a permutation matrix switching the upper $m$ rows and lower $m$ rows of the matrix it operates on, i.e., $\v P\v C=[\v B_{\sss -}^{\top},\v B_{\sss +}^{\top}]^{\top}$. 
Let $\Psi = \mathrm{ran}(\v C) + \mathrm{ran}(\v P\v C)$, whose orthogonal complement is equal to $\Psi^{\perp}= \ker(\v C^{\top}) \cap \ker((\v P\v C)^{\top})$. Moreover, let $\v\Pi_\Psi$ be the orthogonal projection onto $\Psi$ and let $\v z_{\sss \Psi}^{(t)} = \v\Pi_\Psi\v z^{(t)}$ and $ \v z_{\sss \Psi^\perp}^{(t)} = (\v I-\v\Pi_\Psi)\v z^{(t)}$. We then have the following decomposition
\begin{equation}\label{eq:zdecp}
\v z^{(t)}=\v z_{\sss \Psi}^{(t)}+\v z_{\sss \Psi^\perp}^{(t)}.
\end{equation}
It has been proven in \cite{jane2022elsevier} that 
\begin{align*}
\v z_{\sss \Psi^\perp}^{(t)} = \frac{1}{2}\left(\v z_{\sss \Psi^\perp}^{(0)} + \v P\v z_{\sss \Psi^\perp}^{(0)}\right) + \frac{1}{2}(2\theta-1)^t\left(\v z_{\sss \Psi^\perp}^{(0)} - \v P\v z_{\sss \Psi^\perp}^{(0)}\right).
\end{align*}
Thus, for a given graph and $\theta$, $\v z_{\sss \Psi^\perp}^{(t)}$ only depends on the initialization of the  auxiliary variable $\v z^{(0)}$.  The main idea of subspace perturbation is to initialize $\v z^{(0)}$ with a certain distribution having a sufficiently large variance, such that it can help to protect the private data from being revealed to others (see Proposition \ref{prop:share_SubspaceP} for details of the privacy proof). Moreover, by inspecting \eqref{eq:xupc}, we conclude that the $\v x$-update is not affected by $\v z_{\sss \Psi^\perp}^{(t)}$ as $(\v z_{\sss \Psi^\perp}^{(t)} )^{\top}\v C=\v 0$. Hence, the output accuracy is not affected.




\subsection{Adaptive differential quantization} \label{subsec:Quantized}
Adaptive differential quantization schemes in distributed optimization were first introduced in \cite{schellekens2017quantisation,jonkman2018quantisation}. 
They exploit the fact that when the algorithm converges, the difference between successive updates will converge to zero. 
Let $Q^{(t)}\colon \mathbb{R} \rightarrow \{\Delta^{(t)}(a+1/2)\}_{a\in\mathcal{A}}$ denote an $l$-bit uniform (mid-rise) quantization function where  $a\in \mathcal{A}=\{-2^{l-1},-2^{l-1}+1,\ldots,2^{l-1}-1\}$. 
Here $\Delta^{(t)}$ is the quantization cell-width which we define as
\begin{align}\label{eq:dmin}
   \Delta^{(t)} = \max\{\gamma^{t}\Delta^{(0)},\Delta_{\sss \mathrm{min}}\},
\end{align}
where $\gamma \in (0,1)$ is a constant for controlling the decreasing rate, $\Delta^{(0)}$ denotes the initial cell-width and $\Delta_{\sss \mathrm{min}}$ denotes the minimum cell-width.
Given cell-width $\Delta^{(t)}$ at iteration $t$, $Q^{(t)}$ maps each input into its nearest representation value (midpoint of the quantization cells) in $\{\Delta^{(t)}(a+1/2)\}_{\sss a\in\mathcal{A}}$.

Adaptive differential quantization does not operate on the auxiliary variable $\v z$ directly but on the difference of successive variables. That is, 
let $\hat{\v z}$ denote the quantized version of $\v z$ and
define $\Delta z_{\sss j|i}^{(t+1)}$ as
\begin{align} \label{eq:vz}
  \Delta z_{\sss j|i}^{(t+1)}= \left\{\begin{array}{ll} z_{\sss j|i}^{(1)}-z_{\sss j|i}^{(0)} & \textrm{ if $t=0$}\\ z_{\sss j|i}^{(t+1)}-\hat{z}_{\sss j|i}^{(t)} & \textrm{ if $t\geq1$}.
   \end{array}\right.
\end{align}
and
\begin{align}
  \Delta \hat{z}_{\sss j|i}^{(t+1)}= Q^{(t+1)}\left( \Delta z_{\sss j|i}^{(t+1)}\right). \label{eq:vupquant}
\end{align}
Upon receiving the quantized $ \Delta \hat{z}_{\sss j|i}$, each node then calculates $\hat{z}_{\sss j|i}$ as
\begin{align}
  \hat{z}_{\sss j|i}^{(t+1)} =
  \left\{\begin{array}{ll} z_{\sss j|i}^{(0)} +\Delta \hat{z}_{\sss j|i}^{(1)}& \textrm{ if $t=0$}\\ \hat{z}_{\sss j|i}^{(t)}+\Delta \hat{z}_{\sss j|i}^{(t+1)} & \textrm{ if $t\geq1$}. \label{eq:zhatupquant} 
  \end{array}\right.
\end{align}
Thus, the quantized $\hat{\v z}$ can then be constructed as
 \begin{align}\label{eq:tau}
   \hat{\v z}^{(t)}=\v z^{(0)} + \sum_{\tau=1}^{(t)}\Delta \hat{ {\v z}}^{(\tau)}.
 \end{align}
Às such, the quantized auxiliary variable $\hat{\v z}^{(t+1)}$ cannot be reconstructed by the adversary until $\v z^{(0)}$ is known.

The process of quantization will inevitably introduce distortion in the computations. Let $n_{\sss j|i}^{(t+1)}$ denote the introduced error when quantizing $\Delta z_{\sss j|i}^{(t+1)}$. We then have, by combining \eqref{eq:vz} and \eqref{eq:zhatupquant}, that
\begin{align} 
n_{\sss j|i}^{(t+1)} & = \Delta \hat{z}_{\sss j|i}^{(t+1)}-\Delta z_{\sss j|i}^{(t+1)}\nonumber\\
& = \hat{z}_{\sss j|i}^{(t+1)}-z_{\sss j|i}^{(t+1)}. \label{eq:noisezn}
\end{align}

When implementing quantization, i.e., $\Delta z\rightarrow \Delta \hat{z}$, we apply a popular technique in quantization called dithering \cite{schuchman1964dither} when implementing the quantization function $Q^{(t)}$, this ensure that the quantization error
is uniformly distributed over $[-\frac{\Delta^{(t)}}{2},\frac{\Delta^{(t)}}{2}]$, and is independent of $\Delta z_{\sss j|i}^{(t+1)}$, thus of $z_{\sss j|i}^{(t+1)}$. In the coming Section \ref{subsec.delta0} we will show how this property benefits the privacy analysis.

\subsection{Details of the proposed approach}
Putting things together, the proposed approach first applies subspace perturbation and then adopts adaptive differential quantization to quantize the difference variable before updating. More specifically, at $t=0$, subspace perturbation is applied by initializing $z_{\sss i|j}^{(0)}$ at every node $i$ with noise drawn from a distribution having large variance and sending it to neighbors $j\in {\cal N}_i$ via a securely encrypted channel. Each node $i$ then updates its local variable $x_{\sss i}^{(1)}$ and auxiliary variables $z_{\sss j|i}^{(1)}$ according to \eqref{eq:xup} and \eqref{eq:zup}, respectively,
after which $\Delta z_{\sss j|i}^{(1)}$ and $\Delta \hat{z}_{\sss j|i}^{(1)}$ are computed using \eqref{eq:vz} and \eqref{eq:vupquant}, respectively.  After exchanging
these quantities between neighboring nodes, $\Delta \hat{\v z}$ is computed as \eqref{eq:zhatupquant}.

For $t\geq 1$, after achieving the quantized $\hat{z}_{\sss i|j}^{(t)}$, each node $i$ then repeats the following updating steps:
\begin{align}\label{eq:xupL}
  &x_{\sss i}^{(t+1)} =\frac{s_i- \tsum_{\sss j \in {\cal N}_i} B_{\sss i|j}\hat{z}_{\sss i|j}^{(t)}}{1+cd_i} \\
  &\forall j \in {\cal N}_{\sss i}: z_{\sss j|i}^{(t+1)}=\theta \hat{z}_{\sss j|i}^{(t)} +(1-\theta)\left( \hat{z}_{\sss i|j}^{(t)}+2c B_{\sss i|j} x_i^{(t+1)}\right). \label{eq:zupquant}
\end{align}
Note that except transmitting the initialized $\v z^{(0)}$ all the communication channels do not require secure channel encryption when transmitting the quantized $\Delta \hat{\v z}$, since from  \eqref{eq:tau} we can see that the quantized $\hat{\v z}^{(t)}$ cannot be reconstructed unless the initialized $\v z^{(0)}$ is revealed (see Theorem \ref{thm:qsp} for detailed privacy proofs).

Algorithm \ref{alg:pro} summarizes the details of the proposed approach. 
In the coming sections we will analyze the performance of the proposed approach and show its relation to both SMPC and DP.  

\begin{algorithm}[t]
  \caption{Proposed ADQSP: privacy-preserving distributed average consensus via adaptive differentially quantized subspace perturbation}
  \label{alg:pro}
    \SetKwInOut{Input}{Input}
  \SetKwInOut{Output}{Output}
  \SetKwInOut{Initialization}{Initialization}
   \Initialization{Each node $i$ randomly initializes $z_{\sss i|j}^{(0)}$'s from independent Gaussian distributions\footnotemark  ~ ${\cal N}(0;\sigma_z^2)$ for all $j\in \mathcal{N}_i$.}
       \Input{Initialized auxiliary variables $\v z^{(0)}$, quantization parameters $\Delta^{(0)}$, $\gamma, \Delta_{\sss \mathrm{min}}, l$, and $t_{\max}$.
  } 
  \Output{Optimization solution: $x_{\sss i}^{(t_{\max})}$}
  \For{\texttt{$t=0,1,\ldots,t_{\max}-1$}} 
      {
      \If{$t=0$} 
      { 
      Receive $z_{\sss j|i}^{(0)}$ from $j\in {\cal N}_{\sss i}$ via secure channels \cite{dolev1993perfectly}.\\
      $x_{\sss i}^{(1)}\leftarrow$ \eqref{eq:xup}, $\{z_{\sss j|i}^{(1)}\}_{\sss j \in {\cal N}_{\sss i}} \leftarrow$ \eqref{eq:zup}. \
      }
  
      \Else
      { Send $\Delta \hat{z}_{\sss j|i}^{(t)}$ to $j \in {\cal N}_{\sss i}$ via non-secure channels. \\
       $\{\hat{z}_{\sss i|j}^{(t)}\}_{\sss j \in {\cal N}_{\sss i}} \leftarrow$ $ \eqref{eq:zhatupquant}$, $x_{\sss i}^{(t)} \leftarrow$  \eqref{eq:xupL}, \\
       $\{z_{\sss j|i}^{(t+1)}\}_{\sss j \in {\cal N}_{\sss i}}\leftarrow$ \eqref{eq:zupquant}. \
       }
       $\{\Delta \hat{z}_{\sss j|i}^{(t+1)}\}_{\sss j \in {\cal N}_{\sss i}}\leftarrow$ \eqref{eq:vupquant}$, \{\Delta z_{\sss j|i}^{(t+1)}\}_{\sss j \in {\cal N}_{\sss i}}\leftarrow$  \eqref{eq:vz}
      }
\end{algorithm}
\footnotetext{Note that except for being independent, each node can choose its own noise distributions and variances, as long as they are sufficiently large for privacy guarantee.} 
\section{Information-theoretical analysis of existing SMPC and DP approaches} \label{sec.smpc_dp}
Before showing that the proposed ADQSP approach can obtain both SMPC and (relaxed) DP performances, we first introduce and analyze two example approaches of SMPC and DP. In particular, we will give an information theoretical analysis of individual privacy. 
\subsection{Performance analysis of SMPC-based distributed average consensus}\label{sub.smpcave}
\subsubsection{Application of additive secret sharing in distributed average consensus}
Additive secret sharing is a popular SMPC technique. As an example, we now briefly explain how to apply additive secret sharing to distributed average consensus to achieve privacy-preservation \cite{gupta2017privacy,li2019privacyA}. The basic idea is to encrypt the private data before performing averaging using a traditional average consensus algorithm such as Algorithm \ref{alg:ave} or its broadcast alternative.
Let $\mathbb{Z}_p$ be the cyclic group of $p$ elements, represented by the integers $\{0,\ldots,p-1\}$. In order to apply additive secret sharing, we first need to transform all private data $s_i$ to integers in $\mathbb{Z}_p$, where negative numbers are represented by their modular additive inverse and floating point numbers are scaled up into integers.  In addition, $p$ should be sufficiently large such that $p\geq \tsum_{\sss i\in{\cal V}} s_i$.  Hence, discrete random variables are considered here. To apply additive secret sharing \cite{Cramer2015}, each node $i$ first chooses $d_i$ elements $\{r_i^j\in \mathbb{Z}_p\}_{\sss j\in {\cal N}_i}$ uniformly at random. It then sends $r_i^j$ to node $j\in{\cal N}_i$ (which requires secure channel encryption)
 and computes its own share as
\begin{align}\label{eq:last_share}
  r_i^i=s_i-\tsum_{\sss j\in {\cal N}_i}r_i^j ~\mathrm{mod}~ p,
\end{align}
where $r_i^i$ is a noisy version of the private data $s_i$.
Two key properties of additive secret sharing, which we will prove using mutual information in Proposition~\ref{prop.add}, are: 1) the secret can only be reconstructed when all shares are known; 2) no information about the hidden secret can be inferred as long as one share is missing. We have the following result. 
\begin{proposition}\label{prop.add} (Properties of additive secret sharing)
  Let $R_i^j$ be uniformly distributed in $\mathbb{Z}_p$ and let $ r_i^i$ be as in \eqref{eq:last_share}. Furthermore, let $k$ be an element in ${\cal N}_i'={\cal N}_i\cup \{i\}$. Then 
  \begin{align}
    &{\rm I}(S_i;\{ R_i^j\}_{\sss j\in {\cal N}_i'})={\rm I}(S_i;S_i). \label{eq:reconsAddi}\\
        &{\rm I}(S_i;\{ R_i^j\}_{\sss j\in {\cal N}_i'\setminus\{k\}})=0 \label{eq:nminus1Addi}
  \end{align}
\end{proposition}
\begin{proof}
See Appendix~\ref{pf.add}.
\end{proof}
In order to guarantee that adding noise to the original private data has no effect on the average result, each node $i$ adds all received shares $r_j^i$ from neighboring nodes to the share $r_i^i$. That is, node $i$ constructs data
$s'_i$ as
\begin{align}
  s_i'&=r_i^i+\tsum_{\sss j\in {\cal N}_i}r_j^i ~\mathrm{mod}~ p \nonumber\\
  &=s_i+\tsum_{\sss j\in {\cal N}_i}(r_j^i-r_i^j) ~\mathrm{mod}~ p \label{eq:siprim}
\end{align}
and uses $s'_i$ as input to Algorithm \ref{alg:ave} (broadcast type). After obtaining the average of $s'_i$, each node then constructs the final average output, denoted by $s_{\rm smpc}$, as $s_{\rm smpc}= (n\times s'_{\rm ave} ~\mathrm{mod}~ p)/n$.

\subsubsection{Output accuracy}
Since $\tsum_{\sss i\in {\cal V}} \tsum_{\sss j\in {\cal N}_i}(r_j^i-r_i^j)~\mathrm{mod}~ p=0$, the sum of  $s'_i$ is the same as the sum of $s_i$, i.e., 
\begin{align*}
 \tsum_{\sss i\in {\cal V}} s_i' ~\mathrm{mod}~ p &= \tsum_{\sss i\in {\cal V}} \big(s_i+\tsum_{\sss j\in {\cal N}_i}(r_j^i-r_i^j)\big)~\mathrm{mod}~ p\\&=\tsum_{\sss i\in {\cal V}} s_i
\end{align*}
Thus, the MSE is then given by 
\begin{align*}
   e_{\rm smpc}&=(s_{\rm smpc}-s_{\rm ave})^2\\
   &=(\frac{ \tsum_{\sss i\in {\cal V}} s_i' ~\mathrm{mod}~ p} {n}-\frac{\tsum_{\sss i\in {\cal V}} s_i}{n})^2=0.
\end{align*}
Hence, output accuracy is not affected by applying additive secret sharing. 
\subsubsection{Individual privacy}
To analyze privacy, we need to first specify the view of the adversaries. In this analysis, we consider a scenario where each honest node has at least one corrupt neighbor. In that case, the view of the adversaries is given by
  \begin{align}
    {\cal O}_{\sss \rm SMPC,{\cal V}_c} = &\{S_j\}_{\sss j\in {\cal V}_c}\cup \{R_{\sss j}^{k},R_k^j\}_{\sss (j,k)\in \mathcal{E}_c}\cup \{X^{(t)}\}_{t\geq 1} \nonumber.
  \end{align}
With \eqref{eq:xup} and \eqref{eq:zup} we have that 
\begin{align}
   &x_i^{(t+3)}-2\theta x_i^{(t+2)}+(2\theta-1) x_i^{(t+1)} \nonumber \\&
   = \frac{-\tsum_{\sss j\in{\cal N}_i} B_{\sss i|j}( z_{\sss i|j}^{(t+3)}-2\theta z_{\sss i|j}^{(t+2)}+(2\theta-1)z_{\sss i|j}^{(t+1)}) }{1+cd_i}\nonumber\\
    &= \frac{-\tsum_{\sss j\in{\cal N}_i} 2c(1-\theta) ( (1-\theta)x_i^{(t+1)}+\theta x_j^{(t+1)}-x_j^{(t+2)}) }{1+cd_i}. \label{eq:xtplus2}
\end{align}
Hence, the adversary can compute  $\v x^{(t)}$ for $t\geq 3$ using the first two iterations ($t=1,2$).
We have the following result. 

\begin{theorem}\label{thm:additive_privacy} (Information loss of SMPC protocol)
Assume that each honest node has at least one corrupt neighbor. 
Then computing averaging using inputs $s_i'$
given by \eqref{eq:siprim}, the adversaries 
can infer the following information about an arbitrary honest node $i\in {\cal V}_h$:
\begin{align}
& {\rm I}(S_i;{\cal O}_{\sss \rm SMPC,{\cal V}_c}) 
   ={\rm I}(S_i; \tsum_{\sss j \in {\cal V}_{\sss h,1}} S_j).\nonumber 
\end{align}
\end{theorem}
\begin{proof}
See Appendix ~\ref{lmpf.linearSP} and \ref{thmpf.add}.
\end{proof}

Hence, we conclude that the above SMPC approach achieves perfect output accuracy and the individual privacy reaches the bound \eqref{eq:reidealSMPC} in the redefined ideal world, i.e., 
\begin{align}\label{eq:perfSMPC}
  \left( e_{\sss \rm SMPC}=0,{\rm I}(S_i;{\cal O}_{\sss \rm SMPC,{\cal V}_c})={\rm I}(S_i; \tsum_{\sss j \in {\cal V}_{\sss h,1}} S_j)\right).
\end{align}

\subsection{Performance analysis of DP-based distributed average consensus}\label{sec:DP}

In this section, we describe how DP-like methods \cite{kefayati2007secure,huang2012differentially,nozari2017differentially} can be applied in average consensus. Specifically, this approach falls under \emph{local differential privacy} (LDP) \cite{kasiviswanathan2011can}, in which there is no centralized server available and each node considers all other nodes to be corrupt. The most straightforward way to incorporate LDP in distributed average consensus is via \emph{local perturbation}: every node $i$ draws a random noise $r_i$ from some predetermined distribution, and adds it to its private data to obtain perturbed data, denoted as 
\begin{align} \label{eq:dpsr}
  \forall i\in {\cal V}: ~\tilde{s}_i = s_i+r_i
\end{align}
An average consensus protocol is then run on the perturbed data $\tilde{s}_i$ instead of the private data $s_i$.

\subsubsection{Output accuracy}
Since distributed average consensus protocol is performed on $\tilde{s}_i$, its output is then given by $\frac{1}{n}\sum_i \tilde{s}_i$, so the MSE  \eqref{eq:MSE} is given by 
\begin{align}\label{eq:edp}
   {e_{\sss \rm DP}} =\frac{1}{n}\sum_{i=1}^n (\tilde{s}_i-s_i)^2=\frac{1}{n}\sum_{i=1}^n r_i^2. 
\end{align}

\subsubsection{Individual privacy}
As explained before,  mutual information is a relaxed version of $\epsilon$-differential privacy \cite{cuff2016differential,barthe2011information}. For consistency with the rest of this paper, here we also measure privacy leakages of DP approaches via mutual information. Therefore, we refer that our approach achieves performances of relaxed DP approaches.  
DP assumes that all nodes other than the considered node $i$ are corrupt. In particular, the DP adversary knows $\tilde{s}_j$ for $j \neq i$ and, from the output of the average consensus protocol, also $\tfrac{1}{n}\sum_j \tilde{s}_j$. It follows that the DP adversary can always deduce $\tilde{s}_i$, and so individual privacy is given by 
\begin{align} \label{eq.DPpriv}
  {\rm I}(S_i; {\cal O}_{\sss \rm DP,{\cal V}\setminus\{i\}})={\rm I}(S_i;S_i+R_i).
\end{align}
Existing mechanisms give bounds to this, under assumptions of the distribution of the $S_i$. For instance, if it is known that $s_i \in [\alpha,\alpha+M]$ for $\alpha,M \in \mathbb{R}$, then taking $r_i$ to follow the Laplace distribution with parameter $M/\varepsilon$ ensures that \eqref{eq.DPpriv} is bounded by $\varepsilon$ \cite{dwork2006}. Note that all noises $R_i$ are independent of each other and their variances, denoted $\sigma^2_{r_i}$, is given by $\sigma^2_{r_i}=2M^2/\varepsilon^2$. 
 Given \eqref{eq:edp}, the variance of $E_{\sss \rm DP}$ is thus 
\begin{align} 
\frac{1}{n^2}\sum_{i=1}^n \sigma^2_{r_i}= \frac{2M^2}{n\varepsilon^2}.
\end{align}
This shows that the more privacy we demand, i.e., the lower $\varepsilon$ is, the less accurate the output average will be. Hence, there is a \emph{privacy-accuracy trade-off} in DP approaches.

Overall, we conclude the output accuracy and individual privacy of a DP based approach is given by 
\begin{align}\label{eq:perfDP}
    \left( {e_{\sss \rm DP}} =\frac{1}{n}\sum_{i=1}^n r_i^2; \quad {\rm I}(S_i; {\cal O}_{\sss \rm DP,{\cal V}\setminus\{i\}})={\rm I}(S_i;S_i+R_i) \right)
\end{align}

\section{Connection of Proposed approach and existing SMPC and DP approaches} \label{sec.prop2exist} 
We first analyze the output accuracy and individual privacy of the proposed approach and then explain how it connects to existing SMPC and DP approaches. 
\subsection{Output accuracy}
As proved in \cite{schellekens2017quantisation,jonkman2018quantisation,jane2022elsevier}, the output accuracy is dependent on the parameter of $\Delta_{\sss \mathrm{min}}$, i.e., minimum quantization cell-width. 
Let $r_{i, \sss \mathrm{min}}$ denote the residual error in the output of node $i$  compared to the true average $ s_{\sss \mathrm{ave}}$, so that the MSE of the proposed approach is given by 
\begin{align}\label{eq:mseADQSP}
  e_{\sss \rm ADQSP}=\frac{1}{n} \sum_{i\in {\cal V}} r^2_{i, \sss \mathrm{min}}.
\end{align} We have that
\begin{align}
 r_{i, \sss \mathrm{min}}&= x_{\sss i}^{(t)}-s_{\sss \mathrm{ave}}\nonumber\\
 &\stackrel{(a)}{=}\frac{s_i- \tsum_{\sss j \in {\cal N}_i} B_{\sss i|j}\hat{z}_{\sss i|j}^{(t-1)}}{1+cd_i} -s_{\sss \mathrm{ave}}\nonumber \\
 &\stackrel{(b)}{=}\frac{s_i- \tsum_{\sss j \in {\cal N}_i} B_{\sss i|j}z_{\sss i|j}^{(t-1)}}{1+cd_i} -\frac{\tsum_{\sss j \in {\cal N}_i} B_{\sss i|j}n_{\sss i|j}^{(t-1)}}{1+cd_i}-s_{\sss \mathrm{ave}},
  \label{eq:infty}
\end{align}
where (a) uses \eqref{eq:xupL} and (b) uses \eqref{eq:noisezn}. Since the first term in the right-hand side of \eqref{eq:infty} equals $s_{\sss \mathrm{ave}}$ as $t\rightarrow \infty$, we conclude that
\[
r_{i, \sss \mathrm{min}} \rightarrow -\frac{\tsum_{\sss j \in {\cal N}_i} B_{\sss i|j}n_{\sss i|j}^{(t-1)}}{1+cd_i} \quad\text{as } t\to\infty.
\]
Note that for the special case where $\Delta_{\sss \mathrm{min}}=0$, we have that $n_{\sss i|j}^{(t)}\to 0$ as $t\rightarrow \infty$ 
and thus $r^2_{i, \sss \mathrm{min}}\to 0$, hence perfect output accuracy is guaranteed.

\subsection{Individual privacy}


With the adopted quantization scheme, the total collection of information transmitted over the network is $\{z_{\sss i|j}^{(0)}\}_{\sss (i,j)\in \mathcal{E}},\{\Delta \hat{z}_{\sss i|j}^{(t)}\}_{\sss (i,j)\in \mathcal{E}, t\geq 1}$. Of these, only the initialized $\{z_{\sss i|j}^{(0)}\}_{\sss (i,j)\in \mathcal{E}}$ are transmitted over a secure channel. Thus, the eavesdropping adversary has the knowledge of $\{\Delta \hat{z}_{\sss j|k}^{(t)}\}_{\sss (j,k)\in \mathcal{E}, t\geq 1}$
The passive adversary has the knowledge of $\{S_j\}_{\sss j\in {\cal V}_c}\cup\{Z_{\sss j|k}^{(0)},\Delta \hat{Z}_{\sss j|k}^{(t)}\}_{\sss (j,k)\in \mathcal{E}_c, t\geq 1}$ and hence combining the knowledge of these gives that the view of an adversary in our proposed algorithm is
\begin{align*}
  {\cal O}_{\sss \rm ADQSP,{\cal V}_c} = \{S_j\}_{\sss j\in {\cal V}_c} \cup \{Z_{\sss j|k}^{(0)}\}_{\sss (j,k)\in \mathcal{E}_c}\cup \{\Delta \hat{Z}_{\sss j|k}^{(t)}\}_{\sss (j,k)\in \mathcal{E}, t\geq 1} ).
\end{align*}
In what follows we will show that the proposed approach can achieve both  SPMC and DP performances under different parameter settings.

\subsection{Proposed approach achieves SMPC performances by setting $\Delta_{ \mathrm{min}}=0$} \label{subsec.delta0}
We will prove our claim via two main results: 1) the proposed approach obtains similar properties as the SMPC technique, i.e., additive secret sharing; 2) The proposed approach achieves the same performances as the SMPC approach.

We remark that the initialized $ z_{\sss i|j}^{(0)}$'s together with $x_i^{(1)}$ can be considered as shares of $s_i$, similar to the additive secret sharing scheme. In the following, we prove that it also satisfies two key properties of additive secret sharing scheme, similar to Proposition \ref{prop.add}. 
\begin{proposition}\label{prop:share_SubspaceP}
(The proposed ADQSP satisfies two properties required for additive secret sharing.)
Let $Z_{\sss i|j}^{(0)}\sim{\cal N}(0,\sigma_z^2 )$ and let $ x_i^{(1)}=\frac{ s_i- \tsum_{\sss j \in {\cal N}_i} B_{\sss i|j} z_{\sss i|j}^{(0)}}{1+cd_i}$. Furthermore, let $k$ be an element in ${\cal N}_i$ and denote the variance of $S_i$ by $\sigma^2_s$. Then 
  \begin{align}
    {\rm I}(S_i;\{ Z_{\sss i|j}^{(0)}\}_{\sss j\in {\cal N}_i\setminus\{k\}},X_i^{(1)})&\leq \frac{1}{2}\log(1+\frac{\sigma_s^2}{\sigma_z^{2}})\\
    {\rm I}(S_i;\{ Z_{\sss i|j}^{(0)}\}_{\sss j\in {\cal N}_i},X_i^{(1)})&={\rm I}(S_i;S_i).
  \end{align}
  For $\sigma_z^2 \rightarrow \infty$ we obtain 
  \begin{align}\label{eq:inpSP}
    \lim_{\sigma_z^2 \rightarrow \infty}{\rm I}(S_i;\{ Z_{\sss i|j}^{(0)}\}_{\sss j\in {\cal N}_i\setminus\{k\}},X_i^{(1)})= 0.
  \end{align}
\end{proposition}
\begin{proof}
See Appendix~\ref{pf.share_SubspaceP}.
\end{proof}

We now proceed to the main results of individual privacy. 
\begin{theorem} \label{thm:qsp} (Upper and lower bound of the information loss of the proposed ADQSP approach when $\Delta_{\sss \mathrm{min}}=0$.)
Assume that all nodes has at least one corrupt neighbor, the information that the adversaries can infer about an arbitrary honest node $i\in {\cal V}_h$ is upper bounded by: 
\begin{align}\label{eq:spUpper1}
  &{\rm I}(S_i;{\cal O}_{\sss \rm ADQSP,{\cal V}_c})\leq \nonumber
  \\&{\rm I}(S_i;\{S_j -\tsum_{\sss k \in {\cal N}_{\sss j,h}} B_{\sss j|k} Z_{\sss j|k}^{(0)}\}_{\sss j\in {\cal V}_{\sss h,1}}, \{Z_{\sss j|k}^{(0)}-Z_{\sss k|j}^{(0)}\}_{\sss j,k \in {\cal V}_{\sss h,1}} ),
\end{align}
assuming $\sigma_z^2 \rightarrow \infty$, the above becomes 
\begin{align}\label{eq:spUpper2}
  &{\rm I}(S_i;{\cal O}_{\sss \rm ADQSP,{\cal V}_c})\leq {\rm I}(S_i;\tsum_{\sss j\in {\cal V}_{\sss h,1}}S_j).
\end{align}
Assuming $t_{\max} \rightarrow \infty$, the information loss is also lower bounded by: 
 \begin{align}\label{eq:spLower1}
     {\rm I}(S_i; {\cal O}_{\sss \rm ADQSP,{\cal V}_c})
    &\geq I(S_i;\tsum_{\sss j \in {\cal V}_{h,1}} S_j).
 \end{align}
\end{theorem}
\begin{proof}
See Appendix~\ref{thmpf:qspU} for proof of \eqref{eq:spUpper1} and \eqref{eq:spUpper2}, Appendix~\ref{thmpf:qspL} for proof of \eqref{eq:spLower1}.
\end{proof}



Overall,  we conclude that when $\Delta_{ \mathrm{min}}=0$ the output accuracy and individual privacy of the proposed ADQSP protocol is given by 
\begin{align} \label{eq:perfADQSP1}
 \left( e_{\sss \rm ADQSP}=0; {\rm I}(S_i;{\cal O}_{\sss \rm ADQSP,{\cal V}_c})= {\rm I}(S_i;\tsum_{\sss j\in {\cal V}_{\sss h,1}}S_j)
.\right)
\end{align}
which is identical to the performance of SMPC approach \eqref{eq:perfSMPC}. 

\begin{figure*}[ht]
\begin{subfigure}{0.32\textwidth}
\includegraphics[width=0.9\linewidth]{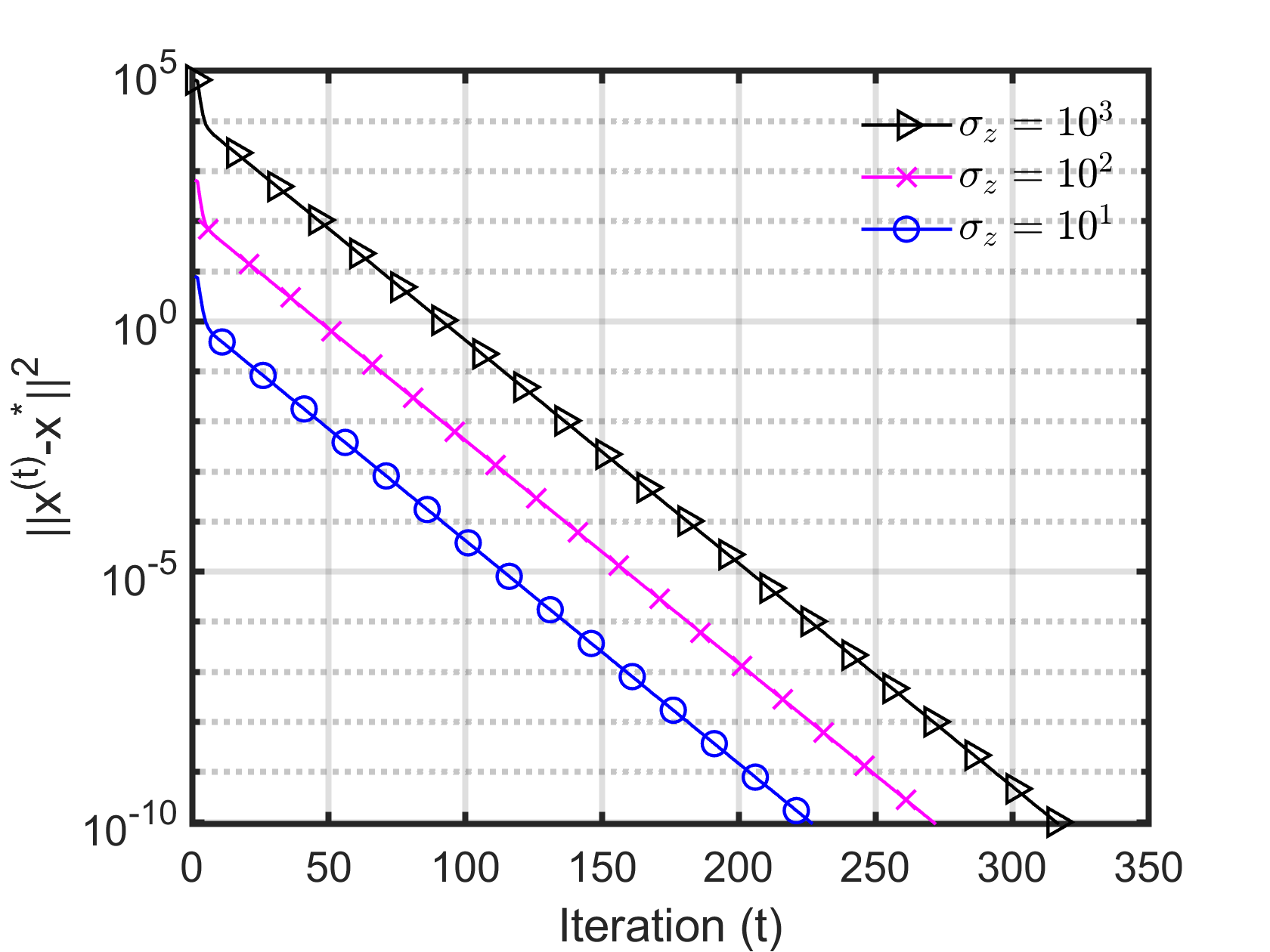} 
\caption{$\theta=0$}
\end{subfigure}
\begin{subfigure}{0.32\textwidth}
\includegraphics[width=0.9\linewidth]{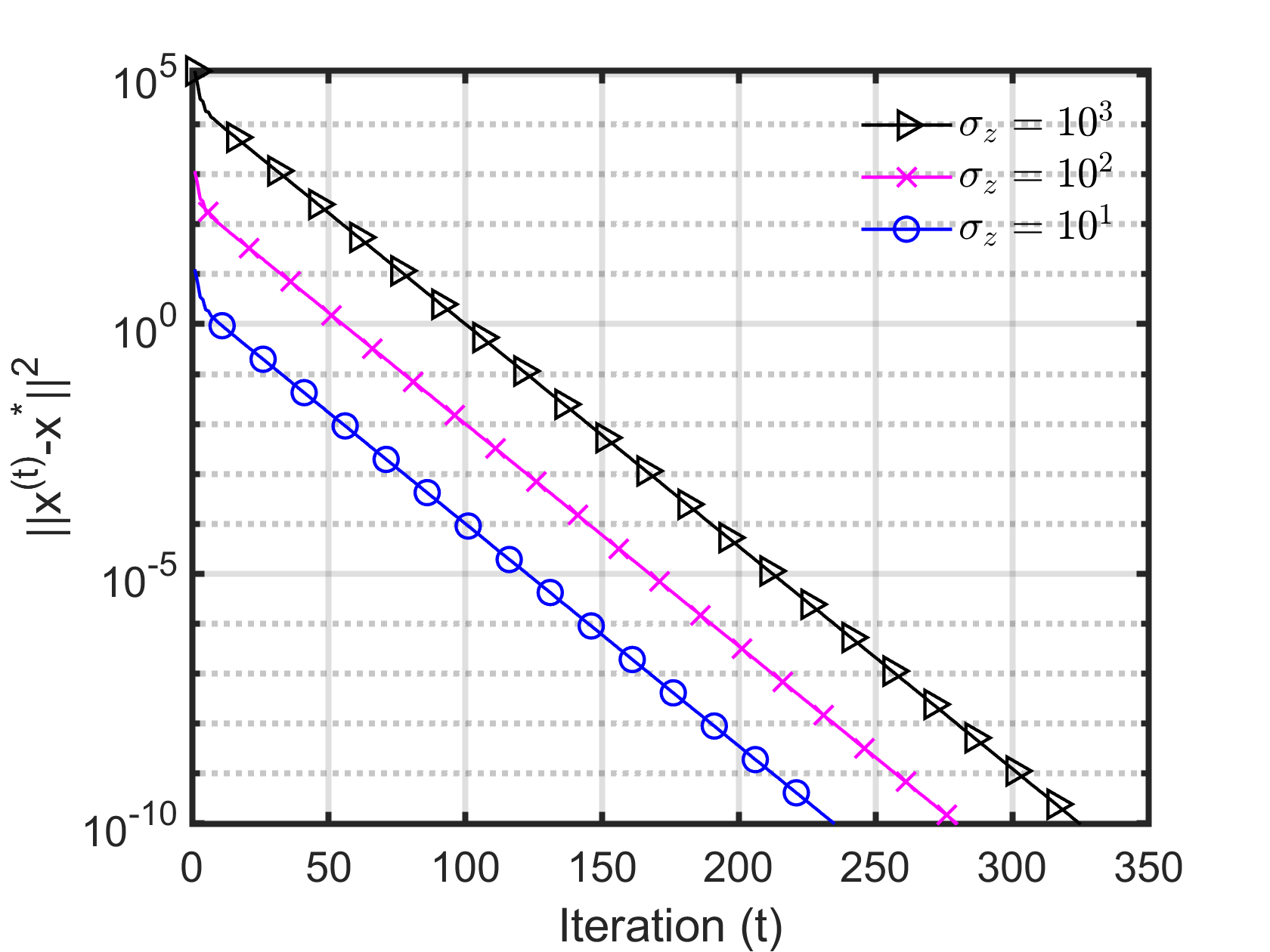}
\caption{$\theta=0.2$}
\end{subfigure}
\begin{subfigure}{0.32\textwidth}
\includegraphics[width=0.9\linewidth]{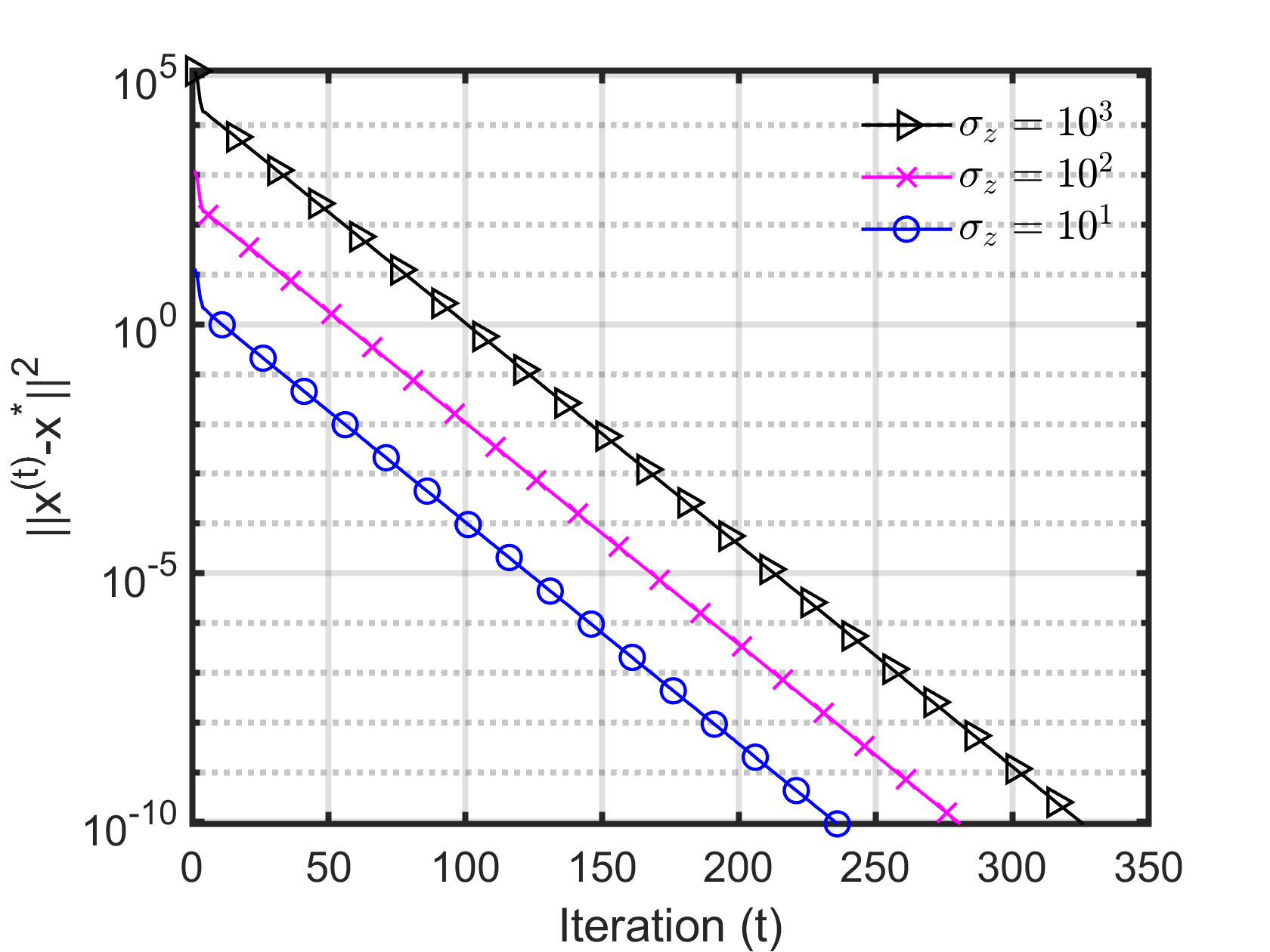}
\caption{$\theta=0.5$}
\end{subfigure}
\caption{Convergence of the optimization variable in terms of three different variances of the auxiliary variable, i.e., $\sigma_z=10^3$, $\sigma_z=10^2$ and $\sigma_z=10^1$ given three different distributed optimizers: (a) $\theta=0$ (PDMM), (b) $\theta=0.2$ and (c) $\theta=0.5$ (ADMM), wherein $\Delta_{\sss \mathrm{min}}=0$.}
\label{fig.conVar}
\end{figure*}

\begin{figure*}[ht]
\begin{subfigure}{0.32\textwidth}
\includegraphics[width=0.9\linewidth]{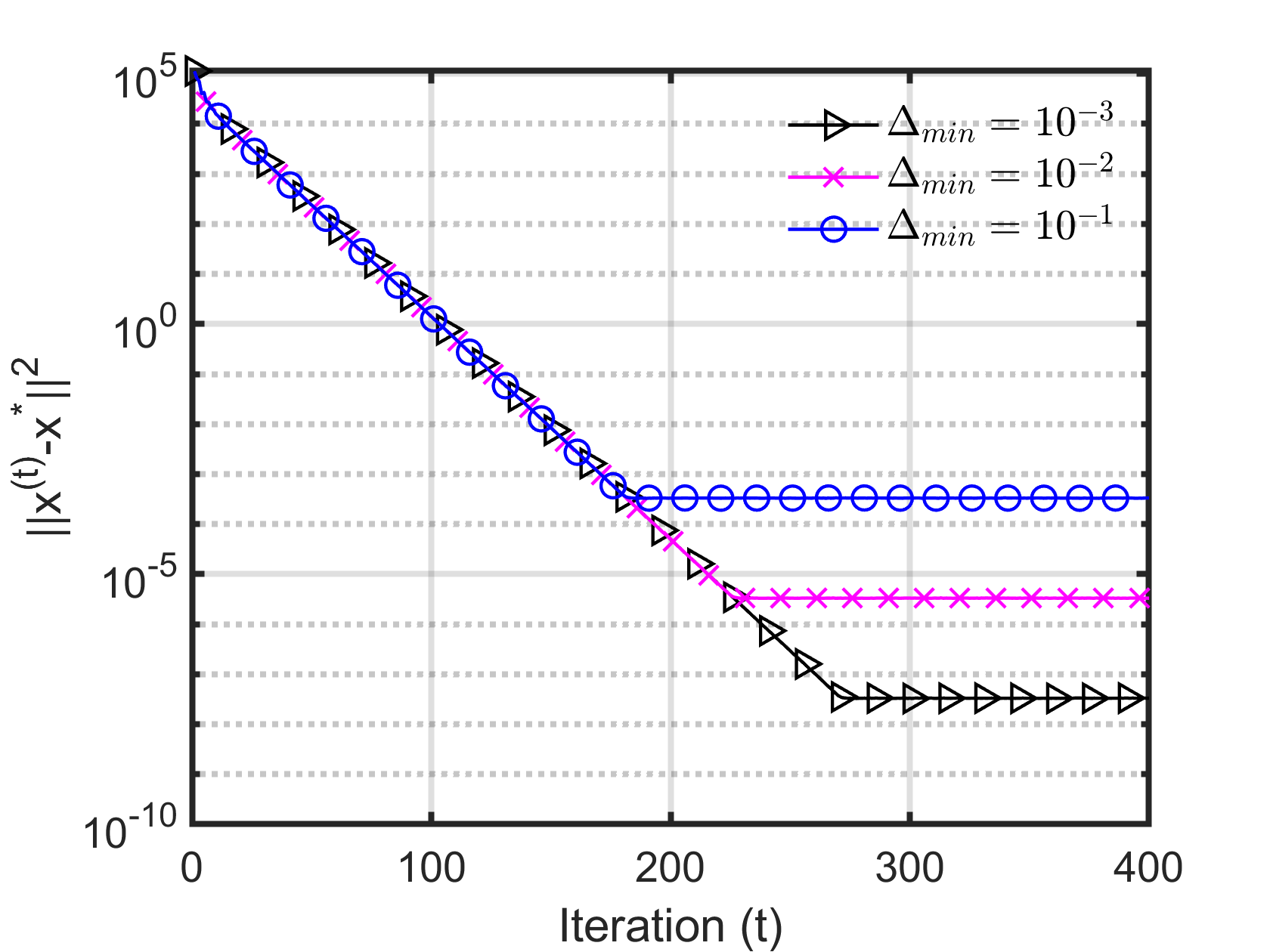} 
\caption{$\theta=0$}
\end{subfigure}
\begin{subfigure}{0.32\textwidth}
\includegraphics[width=0.9\linewidth]{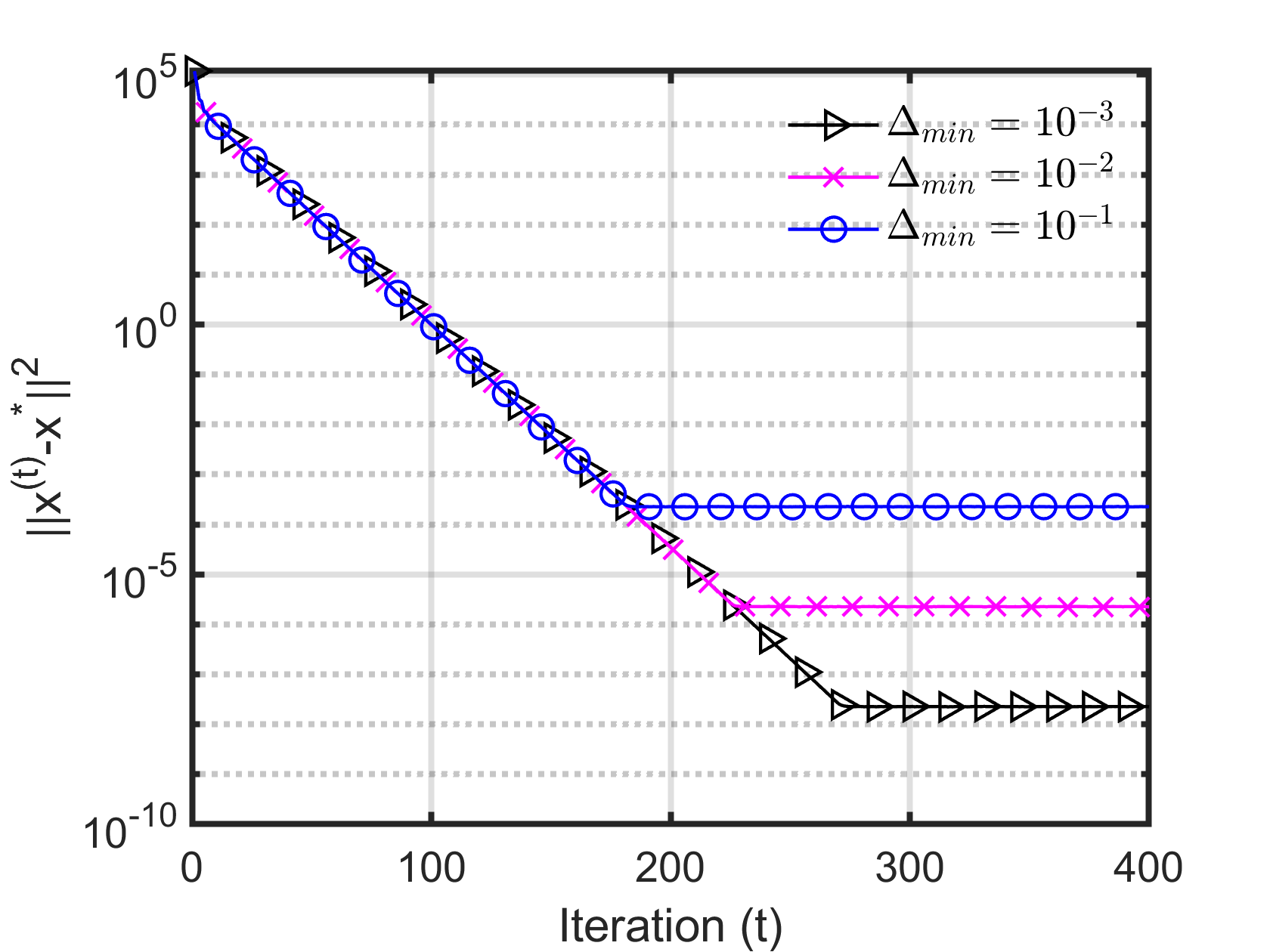}
\caption{$\theta=0.2$}
\end{subfigure}
\begin{subfigure}{0.32\textwidth}
\includegraphics[width=0.9\linewidth]{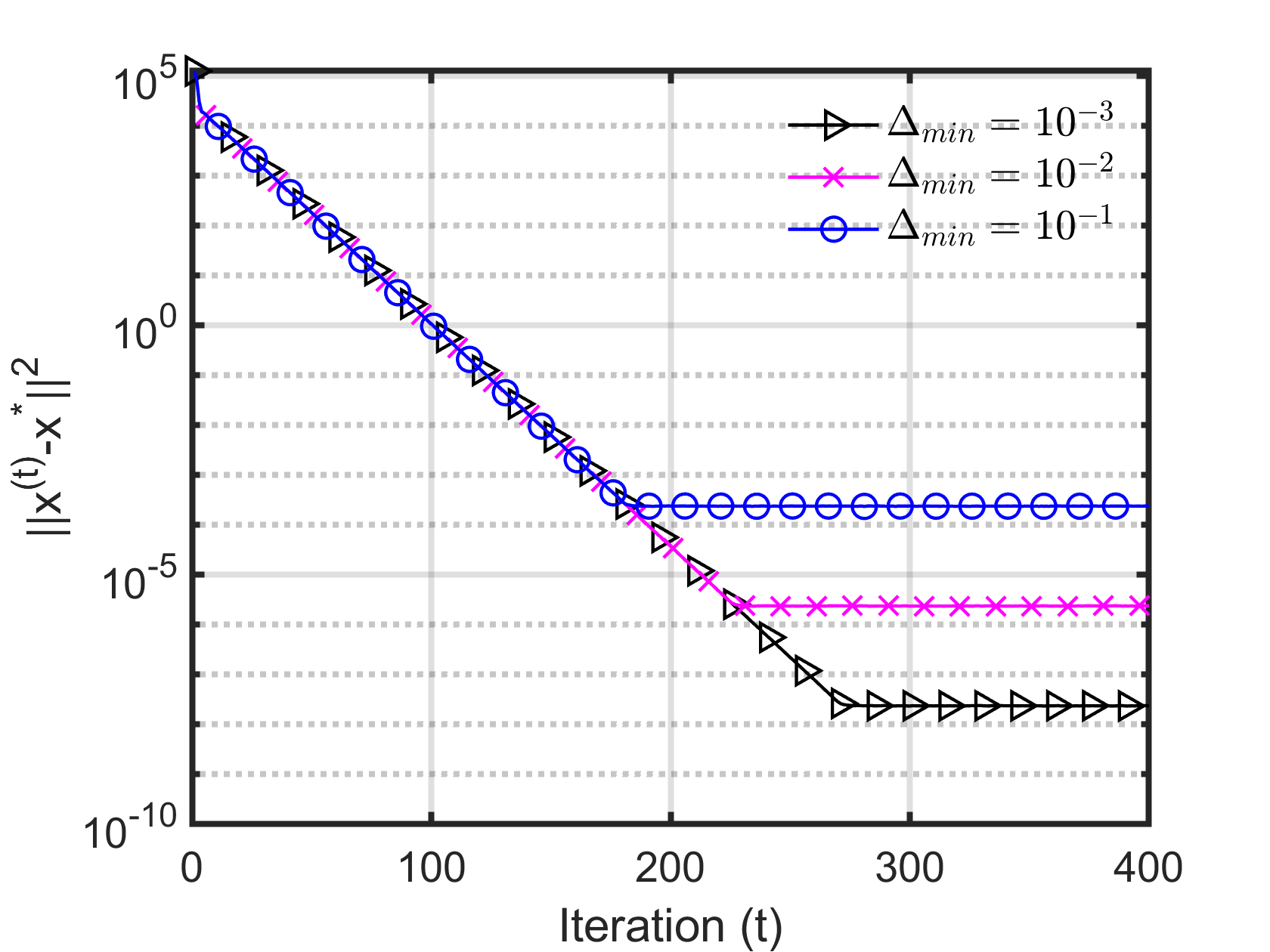}
\caption{$\theta=0.5$}
\end{subfigure}
\caption{Convergence of the optimization variable in terms of three different quantization parameter setting, i.e., $\Delta_{\sss \mathrm{min}}=10^{-3}$, $\Delta_{\sss \mathrm{min}}=10^{-2}$ and $\Delta_{\sss \mathrm{min}}=10^{-1}$ given three different distributed optimizers: (a) $\theta=0$ (PDMM), (b) $\theta=0.2$ and (c) $\theta=0.5$ (ADMM), wherein $\sigma_z=10^3$.}
\label{fig.conDelta}
\end{figure*}

\subsection{Proposed approach achieves performances of (relaxed) DP approaches by setting $\Delta_{ \mathrm{min}}>0$}
DP approach is the worst-case scenario where there are $n-1$ corrupt nodes, i.e., ${\cal V}_c={\cal V}\setminus \{i\}$, which implies ${\cal N}_{\sss i,h}=\emptyset$ and $\mathcal{G}_i=\{i\}$. We have the following result 
\begin{theorem}\label{thm.dp} (Information loss of the proposed approach when ${\cal V}_c={\cal V}\setminus \{i\}$ and  $\Delta_{ \mathrm{min}}>0$)
The information loss is given by 
\begin{align}\label{eq:qNMI}
  {\rm I}(S_i;{\cal O}_{\sss \rm ADQSP,{\cal V}\setminus \{i\}})={\rm I}(S_i;\{S_i+c_{\sss i,k}N_{\sss k|i}^{(t+1)}\}_{t\geq 0})  
\end{align}
 where $c_{\sss i,k} = \frac{1+cd_i}{(1-\theta)2cB_{\sss i|k}}$.
\end{theorem}
\begin{proof}
    See Appendix \ref{thmpf.dp}.
\end{proof}
Note that in this case if $\Delta_{ \mathrm{min}}=0$, we have $N_{\sss k|i}^{(t)} \rightarrow 0$ thus the above ${\rm I}(S_i;{\cal O}_{\sss \rm ADQSP,{\cal V}\setminus \{i\}})={\rm I}(S_i;\{S_i+c_{\sss i,k}N_{\sss k|i}^{(t+1)}\}_{t\geq 0})={\rm I}(S_i;S_i)$ which is maximal, i.e., all the private information is revealed. 

Therefore,  we should set $\Delta_{ \mathrm{min}}>0$ and the corresponding quantization noise  $N_{\sss k|i}^{(t+1)}$ will help to guarantee privacy. 
By inspecting the output accuracy \eqref{eq:mseADQSP} and individual privacy \eqref{eq:qNMI} we can see that increasing $\Delta_{ \mathrm{min}}$ will result less accurate output and more privacy guarantee as the quantization noise $n_{i|j}^{(t)}$s at convergence becomes larger.
This complies to the privacy-accuracy trade-off of DP discussed in Section \ref{sec:DP}. 
Overall, we conclude that by setting $\Delta_{ \mathrm{min}}>0$, the proposed ADQSP gives privacy guarantees in the presence of $n-1$ corrupt nodes, analogous to DP. Its output accuracy and individual privacy are given by
\begin{align} \label{eq:perfADQSP2}
 &\left( e_{\sss \rm ADQSP}=\frac{1}{n} \sum_{i\in {\cal V}} r^2_{i, \sss \mathrm{min}}; \right.
\nonumber \\
 &\left.{\rm I}(S_i;{\cal O}_{\sss \rm ADQSP,{\cal V}\setminus \{i\}})={\rm I}(S_i;\{S_i+c_{\sss i,k}N_{\sss k|i}^{(t+1)}\}_{t\geq 0}) 
.\right)
\end{align}

\section{Numerical results} \label{sec.numRes}
\begin{figure}
  \centering
  \includegraphics[width=.40\textwidth]{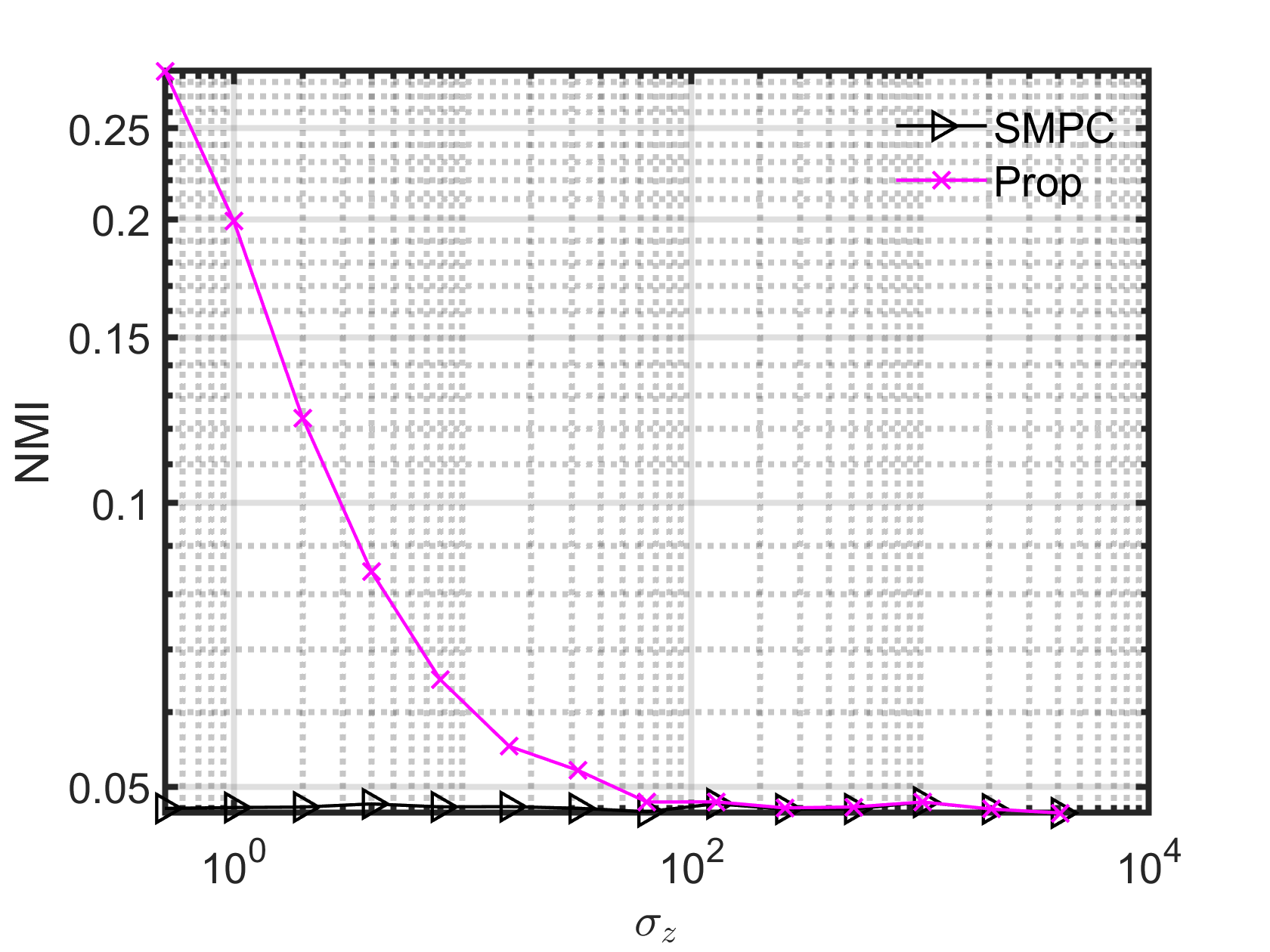}
  \caption{Individual privacy: normalized mutual information (RHS of \eqref{eq:spUpper1}) as a function of $\sigma_z$ using the proposed approach; NMI of ${\rm I}(S_i;\tsum_{\sss j\in {\cal V}_{\sss h,1}}S_j)$ using the existing SMPC approach. }
  \label{fig.nmi_SMPC}
\end{figure}

\begin{figure}
  \centering
  \includegraphics[width=.40\textwidth]{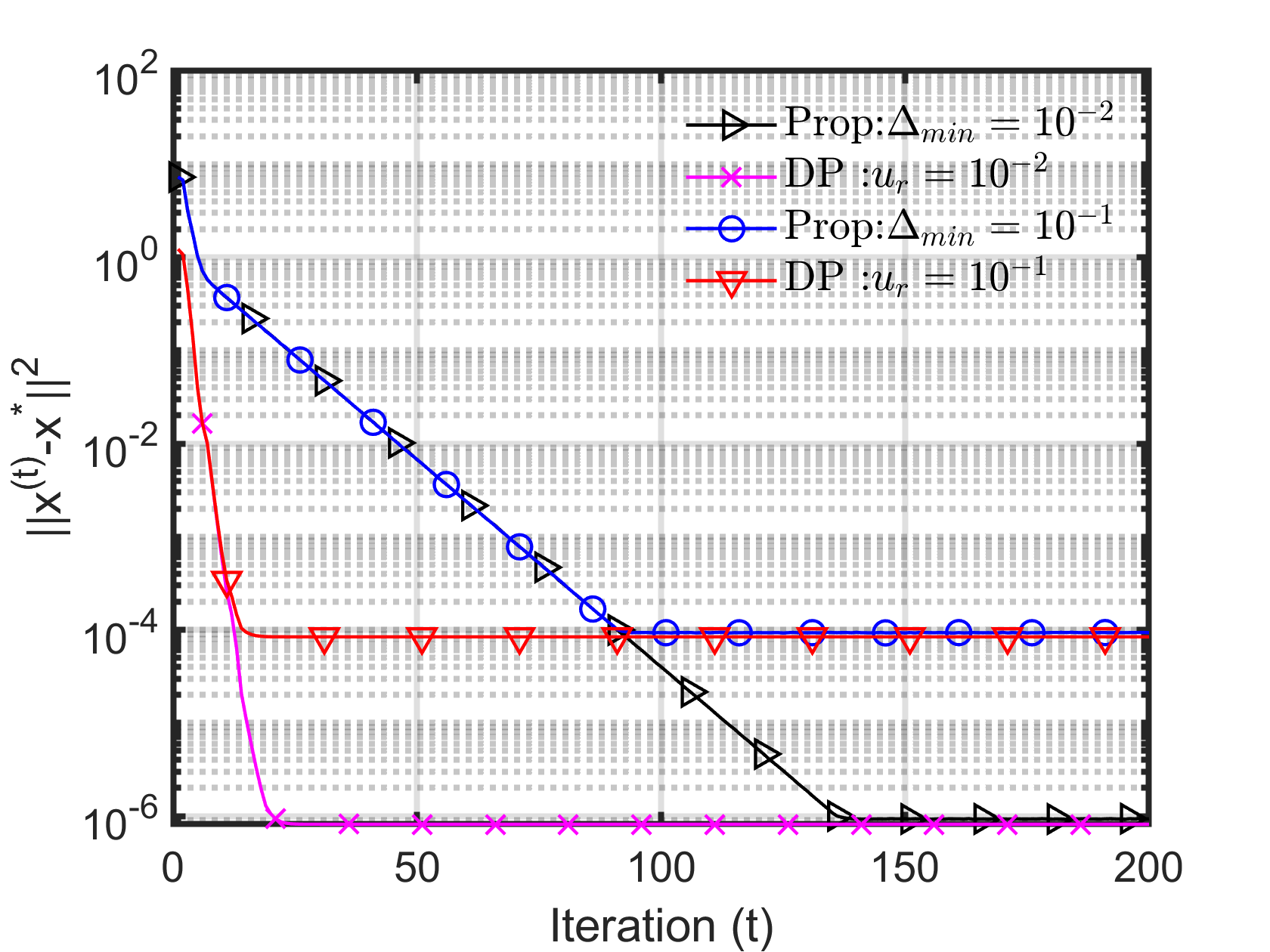}
  \caption{Output accuracy: MSE of the optimization variable in terms of iteration numbers using the proposed approach and DP approach under two different sets of parameter.}
  \label{fig.com_DP}
\end{figure}

\begin{figure}
  \centering
  \includegraphics[width=.40\textwidth]{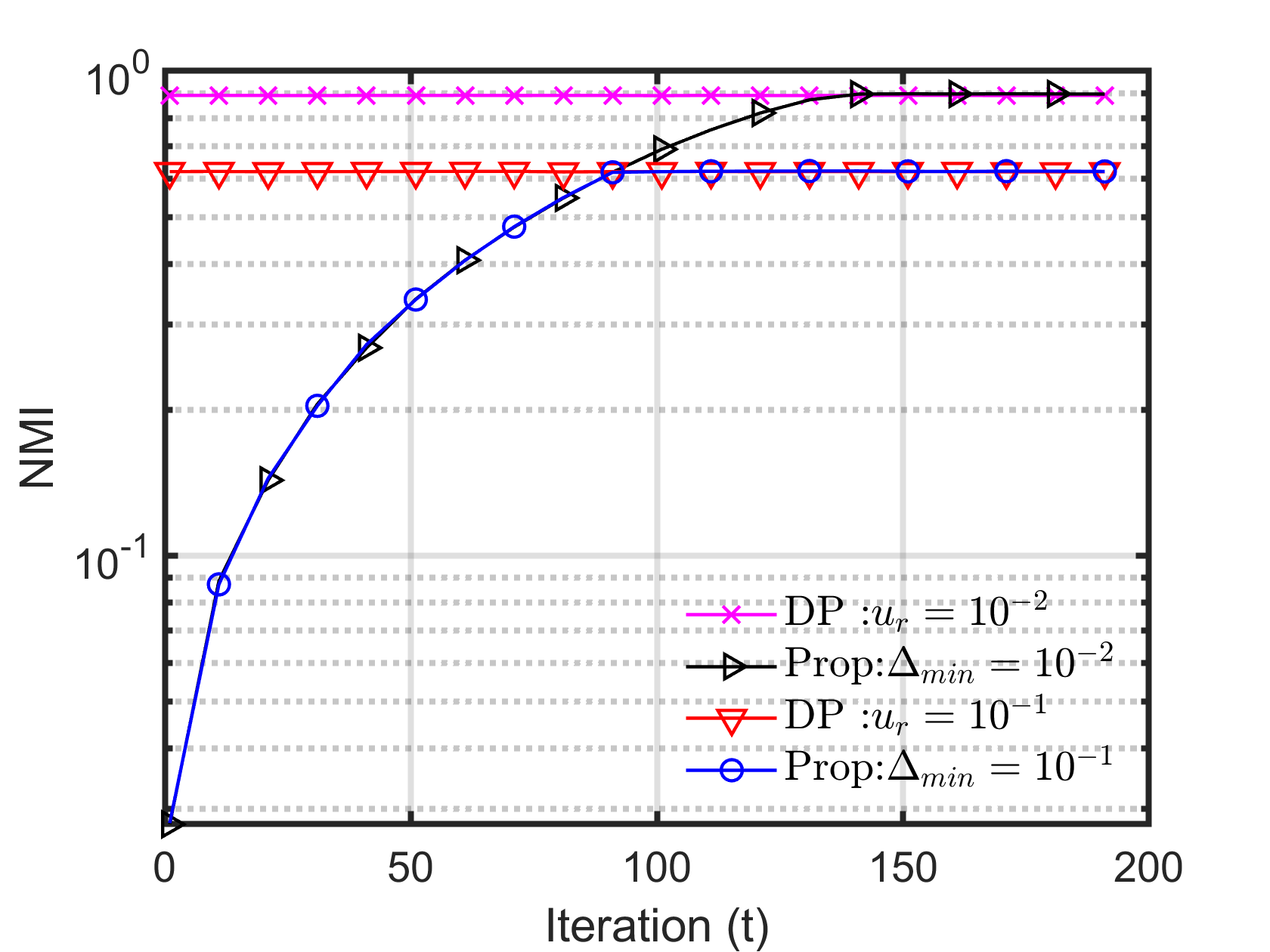}
  \caption{Individual privacy: NMI as a function of iteration numbers using the proposed approach and DP approach under two different set of parameters.}
  \label{fig.nmi_DP}
\end{figure}

In this section, simulation results will be presented to demonstrate the performance of the proposed approach and its connections to both SMPC and DP.

\subsection{Convergence behavior of the proposed approach}
To simulate a distributed network, we simulate a geometric graph with $n=30$ nodes in a room of size $1  \times 1  \times 1 $ meter and the coordinates of each node are randomly generated from uniform distribution within the range of $[0,1]$. Every two nodes are neighbors if and only if their distance is within the radius $\sqrt{\frac{2\log (n)}{n}}$, this guarantees that the generated graph will be connected with a high probability \cite{dall2002random}. Without loss of generality, all private data $\{s_i\}_{i\in {\cal V}}$ are randomly drawn from a Gaussian distribution with unit variance and zero mean. The constant $c$ for controlling the convergence rate is set as $1$. The bitrate $l$ is set as $2$. 

To demonstrate the flexibility and general applicability of the proposed approach, in Fig. \ref{fig.conVar} and \ref{fig.conDelta} we show its convergence behavior under different parameter settings. Each experiment is averaged over $10^4$ times. The mutual information is estimated using the npeet toolbox \cite{ver2000non}. 
\begin{enumerate}
  \item First, we demonstrate that the proposed approach is applicable to different types of distributed optimizers, i.e., to different choices of $\theta\in [0,1)$. We take $\theta=0,~0.2,~0.5$; note that $\theta=0$ corresponds to PDMM and $\theta=0.5$ corresponds to ADMM. As can be seen from Fig.~\ref{fig.conVar} and \ref{fig.conDelta}, the convergence behavior is independent of the choice of $\theta$, and our approach can be applied to different distributed optimizers. 
  \item In Fig. \ref{fig.conVar} we demonstrate the convergence behavior of the optimization variable by varying the variance of the initialized auxiliary variable, which controls the achieved privacy level (see details below in Fig. \ref{fig.nmi_SMPC}). Clearly, we can see that the convergence rate is independent of the variance magnitude, as a higher variance only incurs a higher initial error, thus a larger offset. 
  \item In Fig. \ref{fig.conDelta} we demonstrate the convergence behaviors by varying the quantization cell-width $\Delta_{\sss \mathrm{min}}$. We can see that, as expected, the output accuracy is degraded by increasing $\Delta_{\sss \mathrm{min}}$ (in Fig. \ref{fig.nmi_DP} we will show the privacy leakage becomes less by increasing $\Delta_{\sss \mathrm{min}}$). Hence, the trade-off between privacy and accuracy is controlled by the parameter $\Delta_{\sss \mathrm{min}}$.
\end{enumerate}

\subsection{Connection to existing SMPC and DP approaches}
We now proceed to consolidate our claims that the proposed approaches can reduce to both SMPC and DP by choosing appropriate parameter settings. 
\subsubsection{SMPC}
In Fig. \ref{fig.conVar} we show that by setting $\Delta_{\sss \mathrm{min}}=0$ the proposed approach is able to converge to accurate output average result,  similar to the results reported in SMPC approach \cite{li2019privacyA,gupta2017privacy}. As for individual privacy, we consider the case where there are only two honest nodes in the network and they are connected to each other. In Fig. \ref{fig.nmi_SMPC} we show the achieved privacy level for one honest node of the proposed approach, by taking the normalized mutual information (NMI) on RHS of \eqref{eq:spUpper1} as a function of the variance $\sigma_z$. In addition, we also depict the achieved privacy of the existing SMPC approach, i.e., the normalized mutual information ${\rm I}(S_i;\tsum_{\sss j\in {\cal V}_{\sss h,1}}S_j)$ in the redefined ideal world \eqref{eq:reidealSMPC}. We see that by increasing the variance $\sigma_z$, the proposed approach gets closer to the bound given in SMPC approach. Hence, with Fig. \ref{fig.conVar} and  \ref{fig.nmi_SMPC}, we validate the claim that the proposed approach can achieve the same performances as SMPC approach when $\Delta_{\sss \mathrm{min}}=0$.
\subsubsection{DP}
In order to demonstrate that the proposed approach achieves similar privacy to a DP approach by setting $\Delta_{\sss \mathrm{min}}>0$, we assume that the noise added in DP approach is uniformly distributed over the range $[\frac{-u_r}{2},\frac{u_r}{2}]$. By inspecting Fig. \ref{fig.com_DP} we can see that by setting $u_r=\Delta_{\sss \mathrm{min}}$, the proposed approach achieves similar  output accuracy as the DP  approach. Note that by tuning parameters $c$ the constant for controlling the convergence rate and $\gamma$ the decreasing rate of quantization cell-width, it is possible to achieve a good combination which even has a faster convergence rat compared to the non-quantized case (see \cite{schellekens2017quantisation,jonkman2018quantisation} for more details).
As for individual privacy, let node $i$ be the only honest node and for simplicity we let $c=d_i=1$ and $\theta=0$ such that $c_{\sss i,k} = \frac{1+cd_i}{(1-\theta)2cB_{\sss i|k}}=1$ in \eqref{eq:qNMI}. In Fig. \ref{fig.nmi_DP} we plot the NMI of ${\rm I}(S_i;\{S_i+c_{\sss i,j}N_{\sss j|i}^{(t+1)}\}_{\sss j\in {\cal N}_i})$ as a function of the iteration number $t$ using the proposed approach. For the DP type approach, the NMI is given by ${\rm I}(S_i;S_i+R_i)$. From the figure we conclude that as the iteration grows, ${\rm I}(S_i;\{S_i+c_{\sss i,j}N_{\sss j|i}^{(t+1)}\}_{\sss j\in {\cal N}_i})$ of the proposed approach can attach to the ${\rm I}(S_i;S_i+R_i)$ of DP type approach, which substantiates our theoretical results.

\section{Conclusion} \label{sec.conclu}
In this paper, we proposed a novel privacy-preserving distributed average consensus approach called ADQSP. 
Through a comprehensive information-theoretical privacy analysis, we demonstrate the flexibility and optimality of the proposed approach. Though not directly derived from existing cryptographic tools, it enjoys both the benefits of SMPC and DP techniques and can achieve their performances by controlling the quantization parameter. Experimental results are presented to validate our theoretical analysis.

\appendices
\section{List of Abbreviations}

\begin{table}[t]
  \centering
  \begin{tabular}{l|l}
  \hline
    Abbreviation & Full description \\
    \hline
    SMPC & Secure multiparty computation \\
    DP & Differential privacy  \\
    ADQSP & Adaptive differentially quantized subspace perturbation \\
    MSE & Mean squared error \\
    TTP & Trusted third party \\
    LDP & Local differential privacy\\
   \hline    
  \end{tabular}
  \vspace{3pt}
  \caption{Abbreviations}
  \label{tab:abbreviation}
\end{table}
\section{Broadcast alternative of Algorithm \ref{alg:ave}} \label{appen.broad}
\begin{algorithm}[t]
\caption{Distributed optimization for average consensus (broadcast)}
  At each $i \in {\cal V}$:
  \begin{enumerate}
 \item \label{step.zini} Initialize $\{ z_{\sss i|j}^{(0)}=0\}_{\sss j\in {\cal N}_i}$;
		\item  For $t=0,1,\ldots,t_{\max}-1$ do
		\begin{enumerate}
		  \item Compute $ x_{\sss i}^{(t+1)}$ using \eqref{eq:xup};
		  \item Broadcast $ x_i^{(t+1)}$ to all $j\in {\cal N}_i$;
		  \item Compute $\forall j \in {\cal N}_{\sss i}: z_{\sss i|j}^{(t+1)}, z_{\sss j|i}^{(t+1)}$ using \eqref{eq:zup};
		\end{enumerate}
		\item  Output $ x_i^{(t_{\max})}$
  \end{enumerate}
\end{algorithm}
\section{Proof of Proposition~\ref{prop:smpc}}
 To explain the main idea we use the case that $k_h=2$, i.e., the adversary disconnects the honest nodes into two components denoted as a ``left'' component $\mathcal{L}$ and a ``right'' component $\mathcal{R}$. Let $\v{s}_\mathcal{L}$, $\v{s}_\mathcal{R}$ and $\v{s}_{{\cal V}_c}$ be vectors consisting of the inputs for the left component, the right component, and the corrupt nodes, respectively.
\label{pf.smpc}
\begin{proof}
Let $\operatorname{F}$ be the protocol outputting $\sum_{i=1}^n s_i$. The view of the adversary is
  \begin{align}\label{eq:view_sum_proof}
    {\cal O}_{\sss \operatorname{F},{\cal V}_c}=\{\v{s}_{{\cal V}_c},\v{r}_{{\cal V}_c},\v{m}_\mathcal{L},\v{m}_\mathcal{R}, f(\v{s}_\mathcal{L},\v{s}_{{\cal V}_c},\v{s}_\mathcal{R})\},
  \end{align}
where $\v{r}_{{\cal V}_c}$ is a vector containing the so-called randomness from the corrupt nodes, $\v{m}_\mathcal{L}$ is a vector containing all messages received from nodes in $\mathcal{L}$ and similarly for $\v{m}_\mathcal{R}$  in $\mathcal{R}$. 

The adversary can simulate a run of the protocol where it takes the actions for the nodes in $\mathcal{R}$. It means that it chooses some inputs and randomness and follow the steps in $\operatorname{F}$, where it uses the messages $\v{m}_\mathcal{L}$ when needed in $\operatorname{F}$. If at some point an entry in $\v{m}_\mathcal{R}$ differs from the view it will abort the protocol and start over. Since $\v{s}_\mathcal{R}$ and $\v{r}_\mathcal{R}$ is a valid choice (there might of course be several others) the adversary will succeed at some point meaning that it will find an $\tilde{\v{s}}_\mathcal{R}$ and $\tilde{\v{r}}_\mathcal{R}$ giving the exact same view as in \eqref{eq:view_sum_proof}. Since the adversary uses the exact same messages from the nodes in $\mathcal{L}$ as in the real execution of $\operatorname{F}$ and the correct output is included in the view we must have $f(\v{s}_\mathcal{L},\v{s}_{{\cal V}_c},\v{s}_\mathcal{R})=f(\v{s}_\mathcal{L},\v{s}_{{\cal V}_c},\tilde{\v{s}}_\mathcal{R})$. Hence, the adversary can determine 
  \begin{align*}
    \tsum_{\sss i\in \mathcal{L}} s_i = f(\v{s}_\mathcal{L},\v{s}_{{\cal V}_c},\v{s}_\mathcal{R})-\tsum_{\sss i \in {\cal V}_c} s_i -\tsum_{\sss i \in R} \tilde{s}_i
  \end{align*}
  When the adversary knows $\tsum_{\sss i \in \mathcal{L}}s_i$ it can of course also determine $\tsum_{\sss i \in \mathcal{R}}s_i=\tsum_{\sss i=1}^n s_i-\tsum_{\sss i \in {\cal V}_c}s_i-\tsum_{\sss i \in \mathcal{L}}s_i$.
The above proof can easily be generalized to the case where $k_h>2$ and that in this case the adversary will learn all the sums of the private data in each component. Hence, the proof is now complete.
\end{proof}
\section{Proof of Proposition~\ref{prop.add}}
\label{pf.add}
\begin{proof}
  The first equality holds since 
  \begin{align}\label{eq:add_reconstruction}
  (r_i^i+\tsum_{\sss j\in {\cal N}_i}r_i^j) ~\mathrm{mod}~ p = s_i.
\end{align}
For the last equality we consider two cases. First assume $k=i$, then clearly ${\rm I}(S_i;\{ R_i^j\}_{\sss j\in {\cal N}_i'\setminus\{k\}})=0$ since all the $R_i^j$ are independent of $S_i$. The last case considers $k\neq i$, and in this case we have that 
  \begin{align*}
    &{\rm I}(S_i;\{ R_i^j\}_{\sss j\in {\cal N}_i'\setminus\{k\}})\\
    &\stackrel{\text{$(a)$}}{=}{\rm I}(S_i;\{ R_i^j\}_{\sss j\in {\cal N}_i\setminus\{k\}},S_i-R_i^k)
    \\&\stackrel{\text{$(b)$}}{=}{\rm I}(S_i;S_i-R_i^k),
  \end{align*}
  where $(a)$ uses \eqref{eq:add_reconstruction} and the fact that ${\rm I}(X; Y, Z) = {\rm I}(X; f (Y, Z))$ if $f$ is a bijective function; $(b)$ follows from independence between $\{S_i,S_i-R_i^k\}$ and $\{ R_i^j\}_{\sss j\in {\cal N}_i\setminus\{k\}}$. Due to the fact that $R_i^k$ is uniformly distributed in $\mathbb{Z}_p$, the random variable $S_i-R_i^k$ is also uniformly distributed in $\mathbb{Z}_p$ and hence independent of $S_i$ which completes the proof of \eqref{eq:nminus1Addi}.
\end{proof}

\section{Necessary result for proving Theorem~\ref{thm:additive_privacy}} \label{lmpf.linearSP}
\begin{lemma}\label{lm.linear}
    Let $X_1,\ldots X_n$ and  $R_1,\ldots R_n$ be independent random variables satisfying $\forall i\colon {\rm I}(X_i;X_i+R_i)=0$. Then 
      \begin{align*}
   \forall i:  &{\rm I}(X_i;X_1+R_1,\ldots,X_n+R_n,\tsum_{j=1}^{n} R_j)=I(X_i;\tsum_{j=1}^{n} X_j ).
  \end{align*}
\end{lemma}
\begin{proof}
We first present the following equality:
\begin{align} 
   & \begin{bmatrix}
        1 & 1 & \cdots & 1 & 1\\
        0 & 1 & \cdots & 1 & 1 \\
        \vdots & \ddots & \ddots & \vdots & \vdots\\
        0 & \cdots & 0 & 1 &1\\
        0 & \cdots & \cdots & 0 & 1
    \end{bmatrix}
    \begin{bmatrix}
        X_1+R_1\\
        X_2+R_2\\
        \vdots\\
        X_n+R_n\\
        -\sum_{i=1}^n R_i
    \end{bmatrix} 
    =
    \begin{bmatrix}
        \sum_{i=1}^n X_i\\
        \sum_{i=2}^n X_i- R_i\\
        \vdots\\
        X_n-\sum_{i=1}^{n-1} R_i\\
        -\sum_{i=1}^n R_i
    \end{bmatrix} \label{eq:lnMap} 
\end{align}
With the above result we have 
 \begin{align*}
         &{\rm I}(X_i;X_1+R_1,\ldots,X_n+R_n,\tsum_{j=1}^{n} R_j)\\
         &={\rm I}(X_i;\tsum_{j=1}^{n} X_j, \tsum_{j=2}^{n} X_j-R_1, \ldots,X_n-\tsum_{j=1}^{n-1} R_j,\tsum_{j=1}^{n} R_j)
         \\
         &=I(X_i;\tsum_{j=1}^{n} X_j ).
    \end{align*}
where the first equality holds as the linear map in \eqref{eq:lnMap} is bijective; by inspecting the linear map we can see that the difference of the $k$'th and $(k+1)$'th rows in RHS of \eqref{eq:lnMap} is $X_k+R_k$, which is independent of all $X_i$s and the $k$'th row of \eqref{eq:lnMap}, thus  $X_i\rightarrow \tsum_{j=1}^{n} X_j\rightarrow \tsum_{j=2}^{n} X_j-R_1\rightarrow \ldots\rightarrow X_n-\tsum_{j=1}^{n-1} R_i\rightarrow \tsum_{j=1}^{n} R_j$ forms a Markov chain. Thus, the second equality holds.
\end{proof}

\section{Proof of Theorem~\ref{thm:additive_privacy}}
\label{thmpf.add}
\begin{proof}
We have that 
\begin{align*}
  &{\rm I}(S_i;{\cal O}_{\sss \rm SMPC,{\cal V}_c})\\
&={\rm I}(S_i;\{S_j\}_{\sss j\in {\cal V}_c}\cup \{R_{\sss j}^{k},R_k^j\}_{\sss (j,k)\in \mathcal{E}_c} \cup \{X^{(1)},X^{(2)}\})
\end{align*}
as from \eqref{eq:xtplus2} all $\{X^{(t)}\}_{t\geq 2}$ can be determined using $\{X^{(1)},X^{(2)}\}$. Furthermore, note that since we initialize $z_{\sss i|j}^{(0)}=0$ and  with the inputs from \eqref{eq:siprim} we have that
 \begin{align*}
     x_j^{(1)}&=\frac{ s_j'}{1+cd_j}\\
     x_j^{(2)}&=\frac{ s_j'+ 2c(1-\theta)\tsum_{\sss k \in {\cal V}_j}  x_k^{(1)}}{1+cd_j}
  \end{align*}
from \eqref{eq:xup} and \eqref{eq:zup}.

Thus, there is a bijection between $\{ s_j'\}_{\sss j\in {\cal V}}$ and $\v x^{(1)}$ and furthermore, $\{ s_j'\}_{\sss j\in {\cal V}}$ is enough for computing $ \v x^{(2)}$. In addition, the corrupt $\{ s_j'\}_{\sss j\in {\cal V}_c}$ can be computed using $\{s_j\}_{\sss j\in {\cal V}_c}$ and $\{r_{\sss j}^{k},r_k^j\}_{\sss (j,k)\in \mathcal{E}_c}$ based on \eqref{eq:siprim}. Thus we have
   \begin{align*}
    &{\rm I}(S_i;{\cal O}_{\sss \rm SMPC,{\cal V}_c})\\
    &={\rm I}(S_i;\{S_j\}_{\sss j\in {\cal V}_c},\{R_{\sss j}^{k},R_k^j\}_{\sss (j,k)\in \mathcal{E}_c},\{ S_j'\}_{\sss j\in {\cal V}_h}) \\
    &\stackrel{\text{$(a)$}}{=}{\rm I}(S_i;\{R_{\sss j}^{k},R_k^j\}_{\sss (j,k)\in \mathcal{E}_c},\{ S_j'\}_{\sss j\in {\cal V}_h})
    \\&\stackrel{\text{$(b)$}}{=}{\rm I}(S_i;\{R_{\sss j}^{k},R_k^j\}_{\sss (j,k)\in \mathcal{E}_c},\{ S_j+\tsum_{\sss k\in {\cal N}_{\sss j,h}}R_k^j-R_j^k\}_{\sss j\in {\cal V}_h})\\
    &\stackrel{\text{$(c)$}}{=}{\rm I}(S_i;\{ S_j+\tsum_{\sss k\in {\cal N}_{\sss j,h}}R_k^j-R_j^k\}_{\sss j\in {\cal V}_h})\\
    &\stackrel{\text{$(d)$}}{=}{\rm I}(S_i;\{ S_j+\tsum_{\sss k\in {\cal N}_{\sss j,h}}R_k^j-R_j^k\}_{\sss j\in {\cal V}_{\sss h,1}})\\
    &\stackrel{\text{$(e)$}}{=}{\rm I}(S_i; \tsum_{\sss j \in {\cal V}_{\sss h,1}} S_j)
  \end{align*}
  where (a) holds as the term $\{S_j\}_{\sss j\in {\cal V}_c}$ is independent of all other terms;
  $(b)$ uses the fact that ${\rm I}(X;Y,Z)={\rm I}(X;f(Y,Z))$ if $f$ is a bijective function, $(c)$ holds as all terms in $\{R_{\sss j}^{k},R_k^j\}_{\sss (j,k)\in \mathcal{E}_c}$ are independent of $\{ S_j+\tsum_{\sss k\in {\cal N}_{\sss j,h}}R_k^j-R_j^k\}_{\sss j\in {\cal V}_h}$, and $(d)$ follows from the independence of the different components, and (e) holds as (d) is a special case of Lemma \ref{lm.linear}  wherein $\tsum_{i=1}^{n} R_i=0$. 
\end{proof}

\section{Proof of Proposition~\ref{prop:share_SubspaceP}}
\label{pf.share_SubspaceP}
\begin{proof}
The equality comes immediately as $s_i$ can be determined by $\{ z_{\sss i|j}^{(0)}\}_{\sss j\in {\cal N}_i}$ and $x_i^{(1)}$ using \eqref{eq:xup}. Now consider the inequality. We have that 
\begin{align*}
    &{\rm I}(S_i;\{ Z_{\sss i|j}^{(0)}\}_{\sss j\in {\cal N}_i\setminus\{k\}},X_i^{(1)})\\
    &\stackrel{\text{$(a)$}}{=}{\rm I}(S_i;\{ Z_{\sss i|j}^{(0)}\}_{\sss j\in {\cal N}_i\setminus\{k\}},S_i- B_{\sss i|k}Z_{\sss i|k}^{(0)})\\
    &\stackrel{\text{$(b)$}}{=}{\rm I}(S_i;S_i-Z_{\sss i|k}^{(0)})\\
    &=h(S_i-Z_{\sss i|k}^{(0)})-h(Z_{\sss i|k}^{(0)})\\
    &\stackrel{\text{$(c)$}}{=}h(S_i-Z_{\sss i|k}^{(0)})- \frac{1}{2} \log (2 \pi e \sigma_z^2 )\\
    &\stackrel{\text{$(d)$}}{\leq}\log\left(1+\frac{\sigma_{\sss s}^2}{\sigma_z^2 }\right)
\end{align*}
  where (a) uses $\eqref{eq:xup}$ and the fact that ${\rm I}(X;Y,Z)={\rm I}(X;f(Y,Z))$ if $f$ is a bijective function. (b) follows as $Z_{\sss i|j}^{(0)}$ is independent of $S_i$ and $Z_{\sss i|k}^{(0)}$ and assuming the constant $ B_{\sss i|k}=1$ (for the case of $ B_{\sss i|k}=-1$ the proof is the same.). (c) holds as the entropy of a normal distribution with variance $\sigma^2$ is given by $\frac{1}{2} \log (2 \pi e \sigma^2 )$. (d) follows from the fact that the maximum entropy of a distribution with fixed variance is given by the normal distribution, thus we obtain the above upper bound. 
  
  When $\sigma_z^2\rightarrow \infty$, $\frac{\sigma_{\sss s}^2}{\sigma_z^2}\rightarrow 0$ we thus have 
  \begin{align*}
    \lim_{\sss \sigma_z^2\rightarrow \infty} {\rm I}(S_i;\{ Z_{\sss i|j}^{(0)}\}_{\sss j\in {\cal N}_i\setminus\{k\}},X_i^{(1)})=0,
  \end{align*}
  as mutual information is non-negative.
\end{proof}

\section{Necessary results for proving  Theorem~\ref{thm:qsp}}
\begin{lemma}\label{lm.xr}
Let $X, R$ be independent random variables with variance $\sigma^2_x, \sigma^2_r<\infty$, we have
\begin{align}
  \lim_{\sigma^2_r\rightarrow \infty} {\rm I}(X;X+R)=0, 
\end{align}  
i.e., $X$ is asymptotically independent of $X+R$ if $\sigma^2_r\rightarrow \infty$.
\end{lemma}
\begin{proof}
Let $\gamma = 1/\sigma_r$ and define $R' = \gamma R$. Hence, $R'$ has unit variance. Since mutual information is invariant under scaling, we have
$
I(X ; X+R)=I\left(\gamma X ; \gamma X+R'\right).
$
As a consequence, we have
\begin{align*} \lim _{\sigma^2_r \rightarrow \infty} I(X ; X+R) &=\lim _{\gamma \rightarrow 0} I\left(\gamma X ; \gamma X+R'\right) \\ 
&= I\left(0 ; R^{\prime}\right) =0. 
\end{align*}
\end{proof}

\section{Proof of \eqref{eq:spUpper1} and \eqref{eq:spUpper2} in Theorem~\ref{thm:qsp} }
\label{thmpf:qspU}
\begin{proof}
Consider the difference of two successive $\v x$-updates \eqref{eq:xupc}:
\begin{align*}
  \v x^{(t+1)}-\v x^{(t)}= -(I+c \v C^T\v C)^{-1}(\v C^T\Delta {\v z}^{(t)}).
\end{align*}
Similarly, the difference of two successive $\v z$ updates in \eqref{eq:zupc} is given by
\begin{align}\label{eq:zt}
  \Delta {\v z}^{(t+1)}&= \theta \Delta {\v z}^{(t)}+(1-\theta)\v P \Delta {\v z}^{(t)} \nonumber \\
  &+2c(1-\theta)(\v P\v C (\v x^{(t+1)}-\v x^{(t)})).\nonumber \\
  &= \theta \Delta {\v z}^{(t)}+(1-\theta)\v P \Delta {\v z}^{(t)} \nonumber \\
  &-2c(1-\theta)(I+c \v C^T\v C)^{-1}(\v P\v C\v C^T \Delta {\v z}^{(t)}).
\end{align}
Hence, $\Delta {\v z}^{(t+1)}$ can be determined by $\Delta {\v z}^{(t)}$.

For an honest node $i$ we have
\begin{align}
  &{\rm I}(S_i; {\cal O}_{\sss \rm ADQSP,{\cal V}_c})\nonumber\\
  &={\rm I}(S_i;\{S_j\}_{\sss j\in {\cal V}_c}, \{Z_{\sss j|k}^{(0)}\}_{\sss (j,k)\in \mathcal{E}_c},\{\Delta \hat{Z}_{\sss j|k}^{(t)}\}_{\sss (j,k)\in \mathcal{E}, t\geq 1} ) \nonumber\\
   &\stackrel{\text{$(a)$}}{\leq} {\rm I}(S_i;\{S_j\}_{\sss j\in {\cal V}_c}, \{Z_{\sss j|k}^{(0)}\}_{\sss (j,k)\in \mathcal{E}_c},\{\Delta Z_{\sss j|k}^{(t)}\}_{\sss (j,k)\in \mathcal{E}, t\geq 1} ) \nonumber\\
  &\stackrel{\text{$(b)$}}{=}{\rm I}(S_i;\{S_j\}_{\sss j\in {\cal V}_c}, 
\{Z_{\sss j|k}^{(0)}\}_{\sss (j,k)\in \mathcal{E}_c},\{\Delta Z_{\sss j|k}^{(1)}\}_{\sss (j,k)\in \mathcal{E}} )\nonumber \\
 &\stackrel{\text{$(c)$}}{=}{\rm I}(S_i;\{S_j\}_{\sss j\in {\cal V}_c}, 
\{Z_{\sss j|k}^{(0)}\}_{\sss (j,k)\in \mathcal{E}_c},X^{(1)},\{\Delta Z_{\sss j|k}^{(1)}\}_{\sss (j,k)\in \mathcal{E}_h} ) \nonumber\\
 &\stackrel{\text{$(d)$}}{=} {\rm I}(S_i;\{S_j\}_{\sss j\in {\cal V}_c}, 
\{Z_{\sss j|k}^{(0)}\}_{\sss (j,k)\in \mathcal{E}_c}, X^{(1)}, \{Z_{\sss j|k}^{(0)}-Z_{\sss k|j}^{(0)}\}_{\sss (j,k)\in \mathcal{E}_{\sss h}} ) \nonumber\\
 &\stackrel{\text{$(e)$}}{=} {\rm I}(S_i; \{S_j -\tsum_{\sss  k \in {\cal N}_{\sss j,h}} B_{\sss j|k} Z_{\sss j|k}^{(0)}\}_{\sss j\in {\cal V}_h}, 
\{Z_{\sss j|k}^{(0)}-Z_{\sss k|j}^{(0)}\}_{\sss (j,k)\in \mathcal{E}_{\sss h}} ) \nonumber\\
 &\stackrel{\text{$(f)$}}{=} {\rm I}(S_i;\{S_j -\tsum_{\sss k \in {\cal N}_{\sss j,h}} B_{\sss j|k} Z_{\sss j|k}^{(0)}\}_{\sss j\in {\cal V}_{\sss h,1}}, \{Z_{\sss j|k}^{(0)}-Z_{\sss k|j}^{(0)}\}_{\sss j,k \in {\cal V}_{\sss h,1}} ) \nonumber
\end{align}
where (a) follows from the fact that $\Delta \hat{z}_{\sss j|k}^{(t)}$ is a random function of $\Delta z_{\sss j|k}^{(t)}$ due to quantization, and hence we upper bound the information leakage by replacing $\Delta \hat{z}_{\sss j|k}^{(t)}$ by $\Delta z_{\sss j|k}^{(t)}$. (b) uses the result of \eqref{eq:zt}.
 Since each node has at least one corrupt neighbor thus $\v x^{(1)}$ can be determined using $\{z_{\sss j|k}^{(0)}\}_{\sss (j,k)\in \mathcal{E}_c},\{ \Delta z_{\sss j|k}^{(1)}\}_{\sss (j,k)\in \mathcal{E}} $ through \eqref{eq:zup}. In addition, using $\v x^{(1)}$ and $\{z_{\sss j|k}^{(0)}\}_{\sss (j,k)\in \mathcal{E}_c}$ one can further determine $ \{z_{\sss j|k}^{(1)}\}_{\sss (j,k)\in \mathcal{E}_c} $, thus also $ \{ \Delta z_{\sss j|k}^{(1)}\}_{\sss (j,k)\in \mathcal{E}_c} $. Hence, (c) holds. 
Since $\Delta z_{\sss i|j}^{(1)}=z_{\sss i|j}^{(1)}-z_{\sss i|j}^{(0)}$, replace $z_{\sss i|j}^{(1)}$ using \eqref{eq:zup} we thus have 
\begin{align} \label{eq:difzij0}
  \frac{ \Delta z_{\sss i|j}^{(1)}}{1-\theta}-2c B_{\sss j|i} x_j^{(1)}= z_{\sss j|i}^{(0)}-z_{\sss i|j}^{(0)}.
\end{align}
 Hence, $\{z_{\sss j|i}^{(0)}-z_{\sss i|j}^{(0)}\}_{\sss (i,j)\in \mathcal{E}_{\sss h}}$ and $\{x_j^{(1)}\}_{\sss j\in {\cal V}_h}$ can determine $\{\Delta z_{\sss i|j}^{(1)}\}_{\sss (i,j)\in \mathcal{E}_{\sss h}}$, thus (d) holds. 
In addition, from \eqref{eq:xup} we have 
\begin{align*}
  (1+cd_i)x_{\sss i}^{(1)} +\tsum_{\sss j \in {\cal N}_{\sss i,c}} B_{\sss i|j}z_{\sss i|j}^{(0)} =s_i- \tsum_{\sss j \in {\cal N}_{\sss i,h}} B_{\sss i|j}z_{\sss i|j}^{(0)},
\end{align*}
which again all term in the LHS is known to the adversaries, and hence we obtain (e). For (e) we also remove $x^{(1)}_j$ for corrupt $j$'s since they can be computed from $\{S_j\}_{\sss j\in {\cal V}_c}, \{Z_{\sss j|k}^{(0)}\}_{\sss (j,k)\in \mathcal{E}_c}$ and after removal of this $\{Z_{\sss j|k}^{(0)}\}_{\sss (j,k)\in \mathcal{E}_c}$ is independent of the rest and can be removed. (f) follows by the independence of the different components. Hence, the proof of \eqref{eq:spUpper1} is complete.

If $\sigma_z^2 \rightarrow \infty$, based on Lemma \ref{lm.xr} we have
$i\in {\cal V}_{\sss h,1}: ~ \lim_{\sigma_z^2 \rightarrow \infty}{\rm I}(S_i;S_i -\tsum_{\sss k \in {\cal N}_{\sss i,h}} B_{\sss i|k} Z_{\sss i|k}^{(0)})=0$, i.e., the private data $S_i$ is asymptotically independent of $S_i -\tsum_{\sss k \in {\cal N}_{\sss i,h}} B_{\sss i|k} Z_{\sss i|k}^{(0)}$. Thus, the independence condition required in Lemma \ref{lm.linear} is satisfied, the proof follows similarly, we thus have, 
\begin{align}
 &{\rm I}(S_i;\{S_j -\tsum_{\sss k \in {\cal N}_{\sss j,h}} B_{\sss j|k} Z_{\sss j|k}^{(0)}\}_{\sss j\in {\cal V}_{\sss h,1}}, \{Z_{\sss j|k}^{(0)}-Z_{\sss k|j}^{(0)}\}_{\sss j,k \in {\cal V}_{\sss h,1}} ) \nonumber\\
 &={\rm I}(S_i;\tsum_{\sss j\in {\cal V}_{\sss h,1}}S_j)\nonumber.
\end{align}
Hence, proof of \eqref{eq:spUpper2} is complete.

\end{proof}

\section{Proof of \eqref{eq:spLower1}  in Theorem~\ref{thm:qsp}}
\label{thmpf:qspL}
\begin{proof}
We first derive some equalities which will be used in the proof later.
Since $\Delta z_{\sss j|i}^{(t+1)}=z_{\sss j|i}^{(t+1)}-\hat{z}_{\sss j|i}^{(t)}=\Delta \hat{z}_{\sss j|i}^{(t+1)}-n_{\sss j|i}^{(t+1)}$, replace $z_{\sss j|i}^{(t+1)}$ using \eqref{eq:zupquant} we thus have $\forall t\geq 1$
\begin{align} \label{eq:difzij}
  \frac{ \Delta \hat{z}_{\sss j|i}^{(t+1)}-n_{\sss j|i}^{(t+1)}}{1-\theta}-2c B_{\sss i|j} x_i^{(t+1)}= \hat{z}_{\sss i|j}^{(t)}-\hat{z}_{\sss j|i}^{(t)}.
\end{align}
Thus, for the corrupt neighbor $k\in {\cal N}_{\sss i,c}$ the above becomes
\begin{align} \label{eq:difzij_c}
     x_i^{(t+1)}+\frac{n_{\sss k|i}^{(t+1)}}{2c B_{\sss i|k} (1-\theta)} =\frac{-\hat{z}_{\sss i|k}^{(t)}+\hat{z}_{\sss k|i}^{(t)}}{2c B_{\sss i|k}} +\frac{\Delta \hat{z}_{\sss k|i}^{(t+1)}}{2c B_{\sss i|k} (1-\theta)}.
\end{align}
 
For the honest neighbor $j\in {\cal N}_{\sss i,h}$, using \eqref{eq:tau} the above \eqref{eq:difzij} becomes
\begin{align}\label{eq:difzij_h}
     &x_i^{(t+1)}+\frac{n_{\sss j|i}^{(t+1)}}{2c B_{\sss i|j} (1-\theta)}+\frac{z_{\sss i|j}^{(0)}-z_{\sss j|i}^{(0)}}{2c B_{\sss i|j}} \nonumber\\
     &=\frac{-\tsum_{\tau=1}^{(t)}\Delta \hat{z}_{\sss i|j}^{(\tau)}+\tsum_{\tau=1}^{(t)}\Delta \hat{z}_{\sss j|i}^{(\tau)}}{2c B_{\sss i|j}} +\frac{\Delta \hat{z}_{\sss j|i}^{(t+1)}}{2c B_{\sss i|j} (1-\theta)}.
\end{align}

Using \eqref{eq:tau}, \eqref{eq:xupL} can be written as 
\begin{align}\label{eq:sihc}
  &s_i- \tsum_{\sss j \in {\cal N}_{\sss i,h}} B_{\sss i|j}z_{\sss i|j}^{(0)}-(1+cd_i)x_{i}^{(t+1)} \nonumber\\
  &=\tsum_{\sss j \in {\cal N}_{\sss i,h}} B_{\sss i|j}(\tsum_{\tau=1}^{t}\Delta \hat{z}_{\sss i|j}^{(\tau)})+\tsum_{\sss k \in {\cal N}_{\sss i,c}} B_{\sss i|k}\hat{z}_{\sss i|k}^{(t)},
\end{align}
Note that, all terms in the RHS of the above equations are known by the adversary, i.e., included in ${\cal O}_{\sss \rm ADQSP,{\cal V}_c}$.

As a consequence, take the difference of the  LHS of \eqref{eq:difzij_h} and \eqref{eq:difzij_c} and scale up by $2c$ we have 
\begin{align}\label{eq:mix1}
    \frac{n_{\sss j|i}^{(t+1)}}{B_{\sss i|j} (1-\theta)}+\frac{z_{\sss i|j}^{(0)}-z_{\sss j|i}^{(0)}}{B_{\sss i|j}}-\frac{n_{\sss k|i}^{(t+1)}}{ B_{\sss i|k} (1-\theta)}
\end{align}
In addition, first multiple  $(1+cd_i)$ with the LHS of \eqref{eq:difzij_c} and add to LHS of  \eqref{eq:sihc} we have 
\begin{align} \label{eq:mix2}
    s_i- \tsum_{\sss j \in {\cal N}_{\sss i,h}} B_{\sss i|j}z_{\sss i|j}^{(0)}+c_{\sss i,k}N_{\sss k|i}^{(t+1)},
\end{align}
recall $c_{\sss i,k}=\frac{(1+cd_i)}{2c B_{\sss i|k} (1-\theta)}$.

With the above results we thus have
\begin{align*}
&{\rm I}(S_i; {\cal O}_{\sss \rm ADQSP,{\cal V}_c})\\
&\stackrel{\text{$(a)$}}{=} {\rm I}\left(S_i;\{S_j\}_{\sss j\in {\cal V}_c}, \{Z_{\sss j|k}^{(0)}\}_{\sss (j,k)\in \mathcal{E}_c},\{\Delta \hat{Z}_{\sss j|k}^{(t)}\}_{\sss (j,k)\in \mathcal{E}, t\geq 1}, \right.\\
&\left. \{\frac{N_{\sss j|i}^{(t+1)}}{B_{\sss i|j} (1-\theta)}+\frac{Z_{\sss i|j}^{(0)}-Z_{\sss j|i}^{(0)}}{B_{\sss i|j}}\right. \\
      &\left.-\frac{N_{\sss k|i}^{(t+1)}}{ B_{\sss i|k} (1-\theta)} \}_{\sss i\in {\cal V}_h, j\in {\cal N}_{\sss i,h},k\in {\cal N}_{\sss i,c},t\geq 0}, \right. \\
 & \left.\{S_i- \tsum_{\sss j \in {\cal N}_{\sss i,h}} B_{\sss i|j}Z_{\sss i|j}^{(0)}+c_{\sss i,k}N_{\sss k|i}^{(t+1)} \}_{\sss i\in {\cal V}_h,k\in {\cal N}_{\sss i,c},t\geq 0} \right)\\
 &\stackrel{\text{$(b)$}}{\geq} {\rm I}\left(S_i;
\{\frac{N_{\sss j|i}^{(t+1)}}{B_{\sss i|j} (1-\theta)}+\frac{Z_{\sss i|j}^{(0)}-Z_{\sss j|i}^{(0)}}{B_{\sss i|j}}\right. \\
      &\left.-\frac{N_{\sss k|i}^{(t+1)}}{ B_{\sss i|k} (1-\theta)} \}_{\sss i\in {\cal V}_{h,1}, j\in {\cal N}_{\sss i,h},k\in {\cal N}_{\sss i,c},t\geq 0}, \right. \\
 & \left.\{S_i- \tsum_{\sss j \in {\cal N}_{\sss i,h}} B_{\sss i|j}Z_{\sss i|j}^{(0)}+c_{\sss i,k}N_{\sss k|i}^{(t+1)} \}_{\sss i\in {\cal V}_{h,1},k\in {\cal N}_{\sss i,c},t\geq 0} \right) \\
 &\stackrel{\text{$(c)$}}{\geq} {\rm I}\left(S_i;
  \{\tsum_{\sss j\in {\cal V}_{h,1} }S_j- \tsum_{\sss j,l \in {\cal V}_{h,1}} B_{\sss i|j}(N_{\sss l|j}^{(t+1)}-N_{\sss j|l}^{(t+1)})\right.\\
  &\left.+\tsum_{\sss j\in {\cal V}_{h,1},k \in {\cal N}_{\sss j,c}}c_{\sss j,k}N_{\sss k|j}^{(t+1)}\}_{t\geq 0} \right) \\
   &\stackrel{\text{$(d)$}}{\geq} {\rm I}\left(S_i;
  \tsum_{\sss j\in {\cal V}_{h,1} }S_j- \tsum_{\sss j,l \in {\cal V}_{h,1}} B_{\sss i|j}(N_{\sss l|j}^{(t_{\max})}-N_{\sss j|l}^{(t_{\max})})\right.\\
  &\left.+\tsum_{\sss j\in {\cal V}_{h,1},k \in {\cal N}_{\sss j,c}}c_{\sss j,k}N_{\sss k|j}^{(t_{\max})} \right),
\end{align*}
where (a) uses the fact that \eqref{eq:mix1} and \eqref{eq:mix2} can be computed by the knowledge of the adversary. 
(b) holds by removing the first three terms and consider only the honest component ${\cal V}_{h,1}$. (c) holds by making a linear combination of the terms in (b) where the coefficients are $\frac{1}{2|\mathcal{N}_{\sss i,c}|}$ for the first terms and $\frac{1}{|\mathcal{N}_{\sss i,c}|}$ for the remaining terms. (d) holds because we consider only the last iteration $t_{\max}$.

Remark that if $\Delta_{\sss \mathrm{min}}=0$ then $N_{\sss i|j}^{(t_{\max})}$ will converge almost surely to $0$ when $t_{\max}\rightarrow \infty$. Hence in the limit we have 
\begin{align*}
   {\rm I}(S_i; {\cal O}_{\sss \rm ADQSP,{\cal V}_c})
    &\geq {\rm I}(S_i;\tsum_{\sss j \in {\cal V}_{h,1}} S_j).
\end{align*}
Hence, the proof of \eqref{eq:spLower1} is complete.
\end{proof}

\section{Proof of Theorem \ref{thm.dp}}\label{thmpf.dp}
\begin{proof}
We first present the following equality result: similar to \eqref{eq:zupc}, consider the difference of two successive $\v z$ updates in \eqref{eq:zupquant}:
\begin{align}\label{eq:ztquant}
    &\Delta {\hat{z}_{\sss j|i}}^{(t+1)}+n_{\sss j|i}^{(t+1)}-n_{\sss j|i}^{(t)}\nonumber\\
    &= \theta \Delta {\hat{z}_{\sss j|i}}^{(t)}+(1-\theta)\Delta\hat{z}_{\sss i|j}^{(t)} -\frac{2c(1-\theta) B_{\sss i|j}}{1+cd_i}\sum_{\sss k\in \mathcal{N}_i} B_{\sss i|k}\Delta\hat{z}_{\sss i|k}^{(t)}.
\end{align}
Hence, $\Delta {\hat{z}_{\sss j|i}}^{(t+1)}$ can be determined by the difference of quantization noise over successive iterations $n_{\sss j|i}^{(t+1)}-n_{\sss j|i}^{(t)}$ and $\Delta {\hat{z}_{\sss j|i}}^{(t)}$ and $\Delta {\hat{z}_{\sss i|j}}^{(t)}$ in the previous iteration. 

As for $\Delta {\hat{z}_{\sss j|i}}^{(1)}$ in the first iteration we have 
\begin{align}\label{eq:ztquant0}
    &\Delta {\hat{z}_{\sss j|i}}^{(1)}
    \stackrel{\text{$(a)$}}{=} z_{\sss j|i}^{(1)}-z_{\sss j|i}^{(0)}+n_{\sss j|i}^{(1)} \nonumber\\
    &\stackrel{\text{$(b)$}}{=} (1-\theta)(z_{\sss i|j}^{(0)}-z_{\sss j|i}^{(0)})+\frac{1}{c_{\sss i,j}}(s_i+c_{\sss i,j}n_{\sss j|i}^{(1)}-\tsum_{\sss k\in \mathcal{N}_i} B_{\sss i|k} z_{\sss i|k}^{(0)}),
\end{align}
where (a) uses \eqref{eq:zhatupquant} and \eqref{eq:noisezn}; (b) replaces $z_{\sss j|i}^{(1)}$ using \eqref{eq:xup} and \eqref{eq:zup}. 

We start our proof similar to the proof in Appendix \ref{thmpf:qspL}. However, since ${\cal V}_c={\cal V}\setminus \{i\}$, ${\cal V}_h=\{i\}$ and ${\cal N}_{\sss i,h}=\emptyset$ we do not have any honest neighbors and hence we are not adding \eqref{eq:mix1}. Thus the equality (a) in the proof of Appendix \ref{thmpf:qspL} becomes the first equality below:
 \begin{align*}
& {\rm I}(S_i; {\cal O}_{\sss \rm ADQSP,{ \cal V}\setminus \{i\}})\\
&= {\rm I}\left(S_i;\{S_j\}_{\sss j\in {\cal V}\setminus \{i\}}, \{Z_{\sss j|k}^{(0)},\Delta \hat{Z}_{\sss j|k}^{(t)}\}_{\sss (j,k)\in \mathcal{E}, t\geq 1}\right.,\\
 & \left.\{S_i+c_{\sss i,k}N_{\sss k|i}^{(t+1)} \}_{\sss k\in {\cal N}_{\sss i},t\geq 0} \right) \\
 &\stackrel{\text{$(a)$}}{=}  {\rm I}\left(S_i;\{S_j\}_{\sss j\in {\cal V}\setminus \{i\}}, \{Z_{\sss j|k}^{(0)},\Delta \hat{Z}_{\sss j|k}^{(t)}\}_{\sss (j,k)\in \mathcal{E}, t\geq 1}\right. ,\\
 & \left.\{N_{\sss k|i}^{(t+1)}-N_{\sss k|i}^{(t)}\}_{\sss k\in {\cal N}_{\sss i},t\geq 1}, \{S_i+c_{\sss i,k}N_{\sss k|i}^{(t+1)} \}_{\sss k\in {\cal N}_{\sss i},t\geq 0} \right)\\
  &\stackrel{\text{$(b)$}}{=}  {\rm I}\left(S_i;\{S_j\}_{\sss j\in {\cal V}\setminus \{i\}}, \{Z_{\sss j|k}^{(0)},\Delta \hat{Z}_{\sss j|k}^{(1)}\}_{\sss (j,k)\in \mathcal{E}}\right. ,\\
 & \left.\{N_{\sss k|i}^{(t+1)}-N_{\sss k|i}^{(t)}\}_{\sss k\in {\cal N}_{\sss i},t\geq 1}, \{S_i+c_{\sss i,k}N_{\sss k|i}^{(t+1)} \}_{\sss k\in {\cal N}_{\sss i},t\geq 0} \right)\\
   &\stackrel{\text{$(c)$}}{=} {\rm I}\left(S_i;\{S_j\}_{\sss j\in {\cal V}\setminus \{i\}}, \{Z_{\sss j|k}^{(0)}\}_{\sss (j,k)\in \mathcal{E}}\right. ,\\
 & \left.\{N_{\sss k|i}^{(t+1)}-N_{\sss k|i}^{(t)}\}_{\sss k\in {\cal N}_{\sss i},t\geq 1}, \{S_i+c_{\sss i,k}N_{\sss k|i}^{(t+1)} \}_{\sss k\in {\cal N}_{\sss i},t\geq 0} \right)\\
      &\stackrel{\text{$(d)$}}{=} {\rm I}\left(S_i;\{N_{\sss k|i}^{(t+1)}-N_{\sss k|i}^{(t)}\}_{ \sss k\in {\cal N}_{\sss i},t\geq 1}, \{S_i+c_{\sss i,k}N_{\sss k|i}^{(t+1)} \}_{ \sss k\in {\cal N}_{\sss i},t\geq 1} \right)\\
            &\stackrel{\text{$(e)$}}{=}  {\rm I}\left(S_i; \{S_i+c_{\sss i,k}N_{\sss k|i}^{(t+1)} \}_{\sss k\in {\cal N}_{\sss i},t\geq 1} \right)
 \end{align*}
 where (a) follows as $\{N_{\sss k|i}^{(t+1)}-N_{\sss k|i}^{(t)}\}_{\sss k\in {\cal N}_{\sss i},t\geq 1}$ can be computed by taking the difference of the last term over successive iterations. (b) holds as based on \eqref{eq:ztquant}, all $\{\Delta \hat{Z}_{\sss j|k}^{(t)}\}_{\sss (j,k)\in \mathcal{E}, t\geq 2}$ can be computed using $\{N_{\sss k|i}^{(t+1)}-N_{\sss k|i}^{(t)}\}_{\sss k\in {\cal N}_{\sss i},t\geq 1}$ and  $\{S_j\}_{\sss j\in {\cal V}\setminus \{i\}}, \{Z_{\sss j|k}^{(0)},\Delta \hat{Z}_{\sss j|k}^{(1)}\}_{\sss (j,k)\in \mathcal{E}, t\geq 1}$. (c) holds as $\{\Delta \hat{Z}_{\sss j|k}^{(1)}\}_{\sss (j,k)\in \mathcal{E}}$ can be computed using the other terms based on \eqref{eq:ztquant0}.  (d) 
 holds as $\{S_j\}_{\sss j\in {\cal V}\setminus \{i\}}, \{Z_{\sss j|k}^{(0)}\}_{\sss (j,k)\in \mathcal{E}}$ are independent of the rest two terms. (e) holds as $\{N_{\sss k|i}^{(t+1)}-N_{\sss k|i}^{(t)}\}_{\sss k\in {\cal N}_{\sss i},t\geq 1}$ can be determined by  $ \{S_i+c_{\sss  i,k}N_{\sss k|i}^{(t+1)} \}_{\sss k\in {\cal N}_{\sss i},t\geq 1}$. Hence, the proof is complete.
\end{proof}

\section*{Acknowledgment}
This work was supported in part by the European Union’s Horizon 2020 Research and Innovation Programme under the Marie Skłodowska-Curie under Grant 101008233 and in part by ERC Consolidator under Grant 864075 “Caesar”.

\bibliographystyle{IEEEbib}
\bibliography{dualpath}
\end{document}